%% file: main.tex
\documentclass[acmsmall,screen]{acmart}

\usepackage{adjustbox}

\usepackage{ulem}
\usepackage{caption,subcaption}
\usepackage{svg}
\usepackage{galois}
\usepackage{amsmath}
\usepackage{color,xcolor}
\usepackage{multirow}
\usepackage{tcolorbox}
\usepackage{listings}
\usepackage[ruled,linesnumbered,noend]{algorithm2e}

\SetCommentSty{mycommfont}

\usepackage{pifont}
\usepackage[utf8]{inputenc}
\usepackage{cleveref}
\crefname{section}{§}{§§}
\Crefname{section}{§}{§§}

\usepackage{tcolorbox}
\newcommand{\mybox}[1]{
	\begin{tcolorbox}[
		boxsep=-2.5pt,
		standard jigsaw,
		boxrule=0.6pt,
		opacityback=0,
		sharp corners]
		#1
	\end{tcolorbox}
}

\usepackage{listings}
\lstdefinelanguage{cpp}
{alsoletter={-, =},
	morekeywords={assume, assert, int, i32, i64, u32, u64, while, for, if},
	sensitive=false,
	morecomment=[l]{;}
}
\newcommand{\mytodoblue}[1]{\textcolor{blue}{\ding{46}~{\sf}~#1}}
\newcommand{\mytodored}[1]{\textcolor{red}{\ding{46}~{\sf}~#1}}

\newif\ifshowcomments
\newif\ifshowdeleted
\showcommentsfalse
\showdeletedfalse

\ifshowcomments
\newcommand{\yao}[1]{\mytodored{[yao: #1]}}
\newcommand{\zuo}[1]{\mytodoblue{[zuo: #1]}}
\newcommand{\revise}[1]{\mytodobluetwo{#1}}
\else
\newcommand{\yao}[1]{}
\newcommand{\zuo}[1]{}
\newcommand{\revise}[1]{#1}
\fi

\ifshowdeleted
\newcommand{\delete}[1]{\sout{#1}}
\else
\newcommand{\delete}[1]{}
\fi

\theoremstyle{plain}
\newtheorem{remark}{Remark}[section] 

\newcommand{\ToolName}{InvFinder}

\begin{document}
	
\title{A Fresh Look at Best Inductive Loop Invariant Synthesis for Bit-Vector Relations}
     
\author{Hanrui Zuo}
\orcid{0009-0009-2718-7606}
\affiliation{
  \institution{The State Key Laboratory of Blockchain and Data Security, Zhejiang University}
  \country{China}
}
\email{zarin@zju.edu.cn}        

\author{Peisen Yao}
\orcid{0000-0003-0342-9518}     
\authornote{Corresponding author}
\affiliation{
  \institution{The State Key Laboratory of Blockchain and Data Security, Zhejiang University}      
	\country{China}                   
}
\email{pyaoaa@zju.edu.cn}         

\author{Kui Ren}
\orcid{0000-0002-1969-2591}             
\affiliation{
  \institution{The State Key Laboratory of Blockchain and Data Security, Zhejiang University}       
	\country{China}                   
}
\email{kuiren@zju.edu.cn}         

\begin{abstract}
Synthesizing best inductive invariants (BII) is fundamental to program analysis and verification, yet existing approaches face significant efficiency challenges. We introduce a new formulation for the problem through the lens of mathematical optimization over quantified constraints in first-order theories. 
The formulation offers a constructive and operational
perspective on the BII problem and opens new algorithmic avenues.
Building on this formulation, we present two new algorithms for bit-vector programs: a strategically guided linear search that exploits the lattice structure and a bitwise greedy approach \revise{\delete{whose solver-call count is linear in bit-width (and logarithmic in the size of each row's value space) by carefully selecting pivot elements} that resolves bound bits from high to low with a solver-call count linear in bit-width}.
We evaluate our approach on a comprehensive benchmark suite, demonstrating significant performance improvements over conventional methods based on symbolic abstraction and chaotic iteration. 
Experimental results demonstrate our approach solves up to 86\% more benchmarks than baseline methods, with improved scaling in solver-call count for high bit-widths and improved verification effectiveness when integrated with k-induction.
\end{abstract}

\ccsdesc[500]{Theory of computation~Program analysis}
\ccsdesc[500]{Theory of computation~Abstract interpretation}

\maketitle

\input{1.intro}
\input{2.background}

\input{3.motivation}

\input{4.framework}

\input{5.design}

\input{6.evaluation}

\input{7.related}
\input{8.conclu}

\bibliographystyle{ACM-Reference-Format}
\bibliography{sigproc,synthesis,inv,tmp}

\end{document}

%% file: 1.intro.tex
\section{Introduction}
\label{sec:introduction}
A logical assertion at a program location is \revise{\delete{a} an} \textit{invariant} if it is always satisfied by the values of the program variables whenever the location is reached during program execution. 
The generation of invariants has been key to the proof and analysis of crucial properties, such as
non-interference~\cite{di2008relational,DBLP:journals/pacmpl/ChenWFBD19} and complexity~\cite{Nguyen2012,Alias2010}. As a result, generating invariants has become a cornerstone of program analysis and verification.
Over the years, numerous techniques have been proposed to automate invariant generation, such as abstract interpretation~\cite{cousot1976static,cousot1978automatic,reps2004symbolic,li2014symbolic}, the constraint-based approach~\cite{colon2003linear}, 
IC3/PDR~\cite{cimatti2014ic3},
recurrence analysis~\cite{BreckCKR20,KincaidCBR18,kovacs2004automated,KincaidBBR17}, and machine learning~\cite{si2020code2inv,garg2014ice,sharma2012interpolants,xu2020interval}.

One of the foundational frameworks for deriving invariants is abstract interpretation~\cite{cousot1976static,cousot1978automatic}, which systematically over-approximates behaviors using abstract domains--mathematical structures designed to represent sets of program states. A core concept in abstract interpretation is the Galois connection $\mathcal{C} \galois{\alpha}{\gamma} \mathcal{A}$.
Here, $\alpha$ and $\gamma$ represent the abstraction and concretization functions, which map between the concrete domain $\mathcal{C}$ and the abstract domain $\mathcal{A}$.

The choice of the abstract domain fundamentally limits the precision of abstract interpretation. For a given abstract domain $\mathcal{A}$, the best abstract transformer (BAT) defines the most precise achievable abstraction. The BAT for a concrete transformer $f: \mathcal{C} \to \mathcal{C}$ is the most precise abstract operator that over-approximates the concrete operator. Formally, the BAT is defined as \revise{\delete{$\widehat{f} = \alpha \circ f \circ \gamma$} $f^\star = \alpha \circ f \circ \gamma$}. This equation defines the theoretical precision limit achievable within the abstract domain $\mathcal{A}$. However, it is a non-constructive definition that does not directly yield a practical algorithm for computing BATs.
Prior work has explored BAT computation for finite-height domains~\cite{reps2004symbolic}, template linear domains~\cite{monniaux2009automatic,brauer2010automatic}, and polyhedral domains~\cite{thakur2012method}.

\smallskip
\noindent \textbf{Best Inductive Invariants (BIIs)}.
The best inductive invariant (BII) \cite{thakur2015posthat} for a program point is the most precise over-approximation of the program's reachable states that satisfies the program's inductiveness constraints within the abstract domain.
For example,  Houdini \cite{Flanagan01}  solves a specific version of the BII problem, focusing on inferring conjunctive invariants from a predefined set of predicates.
However, the Houdini paper does not frame its approach in the abstract interpretation framework and does not apply to conventional numerical domains.

Unlike a sufficient invariant, which is tailored to a particular post-condition, the BII is determined only by the transition relation and the abstract domain. This distinction matters for at least three reasons. First, the BII is the theoretical precision ceiling of an abstract domain, and therefore a principled reference point for evaluating domain design. Second, BIIs are query-independent summaries that can be reused across multiple verification clients and properties. Third, stronger auxiliary invariants can improve downstream proof procedures, e.g., $k$-induction, to reduce the required inference depth.

The problem of computing Best Inductive Invariants (BII) has been the focus of significant research and can be broadly classified into two main directions: (1) Work focused on synthesizing the best abstract transformer (BAT) \cite{Graf97,Regehr04,reps2004symbolic,Yorsh04,brauer2011transfer,King10,Monniaux10,thakur2012bilateral,thakur2012method}, which provides a solution to the BII problem; (2)
Research specifically addressing the BII problem itself~\cite{Flanagan01,Yorsh06,Garoche12}, often targeting specific domains.
Most existing techniques for solving the BII problem rely on fixed-point computation methods, in which iterative refinement continues until convergence. Symbolic abstraction~\cite{reps2004symbolic} is often employed to compute the best abstract transformers during these iterations to ensure optimality. However, this approach faces several challenges: computing the best abstract transformers is computationally expensive, and the Kleene-style iterative refinement process can suffer from slow convergence.

\smallskip
\noindent \textbf{This Work}.
We address the computation of the Best Inductive Invariants (BII) within abstract domains over bit-vector arithmetic~\cite{regehr2006deriving,reps2006intermediate,brauer2011transfer}. The domain is particularly well-suited for reasoning about machine integer semantics, including wrap-around behaviors and bitwise operations, with practical applications in analyzing eBPF bytecode, x86 binaries, and hardware designs. Conventional techniques for synthesizing BII often encounter significant efficiency challenges when applied to such domains, as observed in our experiments.

In this paper, we revisit the BII problem and introduce a fundamentally different formulation. Drawing inspiration from constraint-based invariant generation~\cite{colon2003linear,sankaranarayanan2005scalable}, we express BII synthesis as a constrained optimization problem, precisely characterizing BIIs within a well-founded lattice of candidate invariants. This formulation is declarative and parameterized by the abstract domain, without relying on symbolic abstraction over loop-free fragments as a subroutine. Crucially, it is also constructive, serving as the foundation for new algorithmic strategies.

We instantiate this formulation using a \emph{propose-and-refine} framework that performs a directed search over the abstract domain. We introduce two algorithms within this framework: a baseline \textit{linear search} (\cref{sec:framework}) that applies iterative coordinate descent, and an advanced \textit{bitwise greedy} strategy (\cref{sec:design}) inspired by binary lifting~\cite{bender2004level}. The bitwise approach exploits the domain's bit-level structure to asymptotically reduce the worst-case number of solver invocations. Furthermore, we incorporate optimization strategies that leverage intermediate results to construct under-approximations and accelerate convergence.
Using a diverse benchmark suite, we evaluate our algorithms against several baseline methods based on symbolic abstraction~\cite{yao2021program,thakur2012bilateral}. Our approach solves up to 86\% more benchmarks, with the bitwise greedy strategy achieving a speedup of up to $17.9\times$ over \revise{\delete{the best symbolic abstraction algorithm} the bilateral approach by Thakur et al.\cite{thakur2012bilateral}, the primary symbolic-abstraction baseline, which has the better aggregate runtime and generally fewer checks in our reported comparison}. The proposed optimizations reduce solver calls by orders of magnitude, yielding substantial efficiency gains, particularly for higher bit-widths (64-bit and 128-bit variables). An ablation study demonstrates the necessity of each optimization in attaining these improvements. Additionally, integrating our approach with $k$-induction increases verification effectiveness by improving provability while reducing induction depth and verification time.

To summarize, we make the following main contributions:
\begin{itemize}
    \item We introduce a new formulation for best inductive invariant (BII) synthesis that presents BIIs as query-independent, domain-optimal invariants and casts their computation as a constrained optimization problem.
    \item We present two new algorithms for solving the problem applicable to a wide range of abstract domains over bit-vector arithmetic.
    \item We conduct a thorough empirical evaluation comparing our algorithms against symbolic abstraction baselines and show that the inferred BIIs strengthen downstream $k$-induction on a comprehensive benchmark suite. Our tool and benchmarks are available at \url{https://anonymous.4open.science/r/InvFinder-6CA2}.
\end{itemize}

%% file: 2.background.tex
\section{Preliminaries}
\label{sec:preliminaries}

This section first introduces the basic notions of inductive loop invariants and then reviews two families of approaches to invariant inference.

\smallskip 
\noindent \textbf{Inductive Loop Invariant}.
\label{subsec:inv}
For a given loop \texttt{\textbf{while} $G$ \textbf{do} $T$}, the loop invariant inference problem aims to identify a loop invariant $I$ that satisfies
    
\begin{equation}
\frac{Pre \implies Inv \quad \{ Inv \land G \}T\{Inv\} \quad (Inv \land \neg G) \implies Post}{\{Pre\} \texttt{ \textbf{while} $G$ \textbf{do }$T$ } \{Post\}}
\label{eq:loop}
\end{equation}

where $Pre$ is the \textit{pre-condition}, $Post$ is the \textit{post-condition}, and $T$ is the loop body with $G$ as the loop guard (i.e., the condition for loop continuation). The goal is to identify an invariant, $Inv$, that satisfies the following three key properties:
\begin{itemize}
    \item \textit{Init}: The loop invariant, $Inv$, must include all program states reachable by the code executed before the loop. 

    \item \textit{Inductiveness}: For every program state within the loop invariant, after one loop execution, the resulting program state must also satisfy the invariant.
  
    \item \textit{Provability}: When the loop exits (i.e., the guard $G$ evaluates to false), the program states described by the invariant must satisfy the intended post-condition, $Post$. 
\end{itemize}

In the verification community, numerous efforts over the last decade have focused on restricting the search space to invariants that suffice for a fixed post-condition. That direction is appropriate when the objective is to prove a single query as quickly as possible.
In this work, we instead focus on the first two conditions to infer the best inductive invariants. This choice deliberately separates invariant inference from any particular post-condition: the resulting invariant is query-independent and reusable across clients.


\smallskip
\noindent \textbf{\revise{\delete{Abstraction} Abstract} Interpretation for Invariant Inference}.
\label{subsec:absint}
Abstract interpretation is a general framework that enables sound over-approximation of program behaviors. 
Given two complete lattices $(\mathcal{C}, \leq_\mathcal{C})$ and $(\mathcal{A}, \leq_\mathcal{A})$, a pair of functions—an abstraction function $\alpha : \mathcal{C} \to \mathcal{A}$ and a concretization function $\gamma : \mathcal{A} \to \mathcal{C}$—forms a \textit{Galois connection} if, for any element $c \in \mathcal{C}$ and $a \in \mathcal{A}$, the following equivalence holds: $\alpha(c) \leq_\mathcal{A} a \Leftrightarrow c \leq_\mathcal{C} \gamma(a).$

The semantics of a program is defined as the smallest solution of a recursive system of semantic equations $F$. Hence, the abstract program semantics is a
set of states $A$ of a lattice $\langle \mathcal{A}, \sqsubseteq_\mathcal{A}\rangle$
such that $A = F(A)$ where $F$ is a monotone abstract transformer. 
The solution $A$ is
iteratively constructed by $A_{i+1} = A_i \sqcup F(A_i)$, starting
from $A_0 = \bot$. The value $\bot$ denotes the smallest element of
$\mathcal{A}$ and the operation $\sqcup$ denotes the join operation of
$\mathcal{A}$. The sequence $(A_n)$ defines an ascending chain of elements of
$\mathcal{A}$. 
This chain may be infinite, so to enforce the convergence of this sequence, we may need to substitute the operator $\sqcup$ by a \textit{widening operator} $\nabla$ that over-approximates $\sqcup$.

\smallskip
\noindent \textbf{Constraint-based Invariant Inference}.
\label{subsec:template}
To solve the invariant generation problem, 
a large body of work follows the paradigm of the \textit{constraint-based approach}~\cite{colon2003linear,sankaranarayanan2004constraint,sankaranarayanan2005scalable}  also referred to as the template-based approach.
At a high level, the approach rests on the following formulation:
\begin{equation}
\begin{aligned}
  \forall X . Pre(X) &\to Inv(X) \land \\
  \forall X, X' . Inv(X) \land G(X) \land T(X, X') &\to Inv(X'),
\end{aligned}
\label{eq:inv}
\end{equation}
where $X$ denotes the program variables, $X'$ denotes their next-state values, while $G(X)$ and $T(X, X')$ encode the loop guard and the loop body.
To generate invariants, the key idea is to use invariant templates to restrict the search space for $Inv$ and to find invariants that match the template by extracting and solving constraints.
Specifically, the approach involves the following steps.

\smallskip
\noindent \emph{Template Selection}.
The process begins by fixing a parameterized \emph{template} $\mathcal{T}(A, X)$ to characterize the structure of the potential invariant.
Here, $X$ denotes the vector of program variables, and $A$ denotes the vector of unknown parameters to be synthesized, which also constitutes an abstract element in the template domain $\mathcal{D}_T$.
For instance, the template can take the form of a conjunction of inequalities $\bigwedge_{i=1}^m f_i(X) \leq p_i$, where each $f_i$ is a fixed function over $X$ and the bounds $p_i$ constitute the parameters in $A$.
For brevity, we use the notation $A(X)$ to refer to the logical formula $\mathcal{T}(A, X)$, instantiated with the parameter $A$.

\smallskip
\noindent \emph{Constraint Encoding}.
Second, we generate constraints over the template parameters $A$ by enforcing the validation conditions (Eq.~\ref{eq:inv}) on the template instance $A(X)$.
We capture these requirements in a single logical predicate $P(A, X, X')$, which encodes both the initiation and consecution properties:
\begin{equation}
\label{eq:inductiveness_def}
P(A, X, X') \triangleq \big[\mathit{Pre}(X) \to A(X)\big] \land \big[A(X) \land G(X) \land T(X, X') \to A(X')\big].
\end{equation}
Consequently, finding a valid inductive invariant reduces to finding a parameter valuation $A$ that satisfies the formula:
\begin{equation}
\label{eq:template-inv}
\forall X, X' . P(A, X, X').
\end{equation}

\smallskip
\noindent \emph{Constraint Solving}.
Finally, the encoded constraint (Eq.~\ref{eq:template-inv}) can be solved to compute feasible values of the template parameters $A$, which also form a valid inductive invariant in $\mathcal{D}_T$.

%% file: 3.motivation.tex
\section{Problem Formulation}
\label{sec:formulation}

In this section, we first formalize the problem of best inductive invariant (BII) synthesis (\cref{subsec:bii}).
We then examine an existing general framework for solving the problem and discuss its limitations (\cref{subsec:bii:exiting}).
Finally, we present our formulation and outline our algorithmic contributions (\cref{subsec:bii:omt}).

\subsection{The Best Inductive Invariant}
\label{subsec:bii}

This work addresses the problem of computing the \emph{Best Inductive Invariant} (BII) with respect to a specific abstract domain. Informally, the BII represents the strongest possible invariant expressible within the domain that satisfies the program's inductive properties.

\begin{definition}[Best Inductive Invariant]
    \label{def:bii}
    For a given abstract domain $\mathcal{A}$ and the program encoded by $P(A,X,X')$ (Eq.~\ref{eq:inductiveness_def}), an abstract element $A^\star \in \mathcal{A}$ is the \emph{Best Inductive Invariant} if and only if:
    (1) \emph{Validity:} $A^\star$ is a valid inductive invariant: $\forall X, X' . \; P(A^\star, X, X')$;
    (2) \emph{Optimality:} $A^\star$ is the strongest among all valid \revise{inductive} invariants in $\mathcal{A}$: $\forall A \in \mathcal{A} . \; (\forall X, X' . \; P(A, X, X')) \implies A^\star \sqsubseteq A$.
\end{definition}

In other words, $A^\star$ is the \emph{least element} (in terms of the lattice order $\sqsubseteq$) of the set of all valid inductive invariants in $\mathcal{A}$.
The computation of the BII is of both theoretical and practical significance.  From a theoretical perspective,  it establishes the precision ceiling for any analysis constrained by a given abstract domain. In practice, BIIs enable the derivation of precise summaries for loops, functions, and other program constructs, which are critical for verification, bug detection, and other downstream tasks.

\begin{remark}
We would like to emphasize that our definition differs from the concept of ``strongest (polynomial) invariants'' commonly used in research on computability results, such as~\cite{karr76,OS04ICALP,joel2018polynomial,hrushovski23jacm,mullner24popl}.
These studies typically focus on computing \emph{sets of} affine or polynomial \emph{equalities} that serve as loop invariants. For a comprehensive overview of these related results, we recommend consulting~\cite{mullner24popl}. 
\end{remark}

\begin{example}
    The original predicate abstraction algorithm~\cite{Graf97} can be viewed as an instance of BII synthesis. 
    Given a set of predicates $S = \{P_1, \ldots, P_n\}$, each element in the abstract domain represents a Boolean formula over these predicates. Their fixed-point algorithm computes the most precise inductive invariant expressible in this domain.
\end{example}

\subsection{Existing Approach to BII Synthesis}
\label{subsec:bii:exiting}

In the most basic approach to solving the best inductive invariant (BII) problem, we assume we have a standard fixed-point solver performing chaotic iteration.
Compared to standard equation solvers, the basic idea is to adopt the best abstract transformers to improve the precision of each iteration and, ultimately, the precision of the inferred invariant.

\begin{algorithm}[t]
	\caption{Best invariant generation via chaotic iteration and  symbolic abstraction}
    \label{alg:bii:symabs}
	\KwIn{A transition system and abstract domain $\mathcal{A}$}
	\KwOut{The best inductive invariant $A^\star \in \mathcal{A}$}
	$A \gets \alpha(\{X|Pre(X)\})$\;
	\While{\revise{\delete{$A \neq \widehat{f}(A)$} $A \neq A \sqcup f^\star(A)$}}{
        \revise{\delete{$A \gets A \sqcup \widehat{f}(A)$} $A \gets A \sqcup f^\star(A)$}\tcp*[l]{apply the symbolic abstraction operator \revise{\delete{$\widehat{f}(A)$} $f^\star(A)$}}
	}
	\Return{$A$}\;
\end{algorithm}

\smallskip
\noindent \textbf{Best Abstract Transformer}. Given a concrete transformer $f : \mathcal{C} \to \mathcal{C}$, the \textit{best abstract transformer} \revise{\delete{$\widehat{f} : \mathcal{A} \to \mathcal{A}$} $f^\star : \mathcal{A} \to \mathcal{A}$} that over-approximates $f$ is defined as \revise{\delete{$\widehat{f} = \alpha \circ f \circ \gamma : \mathcal{A} \to \mathcal{A}.$} $f^\star = \alpha \circ f \circ \gamma : \mathcal{A} \to \mathcal{A}$.}
This is the most precise sound abstraction of $f$ in the abstract domain $\mathcal{A}$ since, for any other sound abstraction $f^\#$, it holds that \revise{\delete{$\widehat{f}(A) \leq_\mathcal{A} f^\#(A)$} $f^\star(A) \leq_\mathcal{A} f^\#(A)$} for all $A \in \mathcal{A}$.
However, \revise{\delete{this definition is limited for two reasons} two distinct practical issues remain}:
\begin{itemize}
    \item It is \textit{non-constructive}, meaning it does not necessarily provide an algorithm to (1) compute an explicit representation of the most precise transfer function \revise{\delete{$\widehat{f}$} $f^\star$} or (2) to apply the function to abstract states and obtain the results.
    \item \revise{Best abstractions are not compositional in general.} The composition of the best abstractions of two functions $f$ and $g$ does not always yield the best abstraction of their composition $f \circ g$, a known limitation of abstract interpretation.
\end{itemize}

Symbolic abstraction~\cite{reps2004symbolic} provides a practical mechanism for computing best abstract transformers. Given a formula $\varphi$ that represents the concrete semantics and an abstract domain $\mathcal{A}$, symbolic abstraction computes the best approximation of $\varphi$ as an element in $\mathcal{A}$ (i.e., the strongest consequence of a formula $\varphi \in \mathcal{L}$ expressible within an abstract domain $\mathcal{A}$.)
Depending on the context, the formula $\varphi \in \mathcal{L}$ may encode different language constructs, such as a concrete transformer for an instruction, a basic block, or a loop-free program fragment.

\begin{remark} The existence of best abstract transformers is not guaranteed for all abstract domains~\cite{cousot1978automatic,cousot1995formal,cousot2011logical}. In such cases, only one direction of the Galois connection may be maintained. Our work focuses exclusively on abstract domains where the best abstract transformers are well-defined. 
\end{remark}

\smallskip
\noindent \textbf{Synthesizing BII via Symbolic Abstraction}. 
Algorithm~\ref{alg:bii:symabs} illustrates the standard procedure for computing BIIs through symbolic abstraction within a fixed-point iteration framework. For abstract domains of finite height, Kleene iteration without widening is sufficient, generating increasingly precise under-approximations until a fixed point is reached.

\begin{theorem}[\revise{Correctness and complexity of Algorithm~\ref{alg:bii:symabs}}]
\label{thm:symabs-correct}
\revise{\delete{For an abstract domain with no infinite ascending chains, Algorithm~\ref{alg:bii:symabs} is guaranteed to terminate and compute the best inductive invariant expressible in the domain~\cite{thakur2014symbolic}.} Assume that $f^\star$ is the exact best abstract transformer and that the part of $\mathcal{A}$ reachable from $\alpha(\mathit{Pre})$ has finite height $h$. Algorithm~\ref{alg:bii:symabs} makes at most $h$ strict updates, and returns the BII expressible in $\mathcal{A}$~\cite{thakur2014symbolic}.}

\end{theorem}

\begin{proof}
\revise{Let $A_0=\alpha(\mathit{Pre})$ and $A_{n+1}=A_n\sqcup f^\star(A_n)$. Monotonicity of $f^\star$ makes $(A_n)_n$ ascending. For any inductive element $I$, initiation gives $A_0\sqsubseteq I$, and $f^\star(I)\sqsubseteq I$; induction therefore gives $A_n\sqsubseteq I$ for every $n$. Finite height ensures stabilization after at most $h$ strict updates. At stabilization, $f^\star(A)\sqsubseteq A$, so $A$ is inductive; because $A\sqsubseteq I$ for every inductive $I$, it is the BII by Definition~\ref{def:bii}.}
\end{proof}


\begin{example}
The predicate abstraction algorithm in~\cite {Graf97} implements this approach by computing, at each iteration, the strongest consequence expressible using the given predicates. Various optimizations have been proposed to accelerate symbolic abstraction~\cite {lahiri2006smt}, but the general framework remains unchanged.
\end{example}

Despite its theoretical elegance, the symbolic abstraction-based framework has several limitations: (1) the prohibitive expense of computing best abstract transformers, especially for complex domains; (2) the potential slow convergence of Kleene-style iterations; and (3) limited resilience to timeouts, i.e., returning non-trivial results (e.g., non-$\top$)  when computation cannot be completed.

\begin{remark} Thakur et al.~\cite{thakur2012bilateral} proposed an ``anytime algorithm'' for symbolic abstraction, producing a nontrivial result in a timeout. 
However, the overall invariant inference process still requires a fixed-point iteration, and generating a nontrivial invariant remains challenging for timeout resilience. \end{remark}

\subsection{Problem Statement}
\label{subsec:bii:omt}
\noindent \textbf{Constraint-based BII Synthesis}.
This work revisits the problem of Best Inductive Invariant (BII) synthesis from a constraint-optimization perspective. Inspired by constraint-based program analysis~\cite{colon2003linear,sankaranarayanan2004constraint,sankaranarayanan2005scalable}, we recast the BII problem as a direct optimization problem over the abstract domain $\mathcal{A}$ and the partial order $\sqsubseteq$ of its abstract elements:

\mybox{
We define the \emph{lattice rank} of an element $A \in \mathcal{A}$, denoted $rank_{\mathcal{A}}(A)$, as the length of the longest ascending chain from the bottom element $\bot$ to $A$. The BII synthesis problem is defined as:

\begin{align*}
  \textbf{Minimize } & \quad rank_{\mathcal{A}}(A) \\
  \textbf{Subject to } & \quad \forall X, X' . \; P(A, X, X'),
\end{align*}

where $P(A, X, X')$ denotes the inductiveness condition defined in Eq.~\ref{eq:inductiveness_def}.
}

This formulation characterizes the BII as the most precise (lowest-ranked) valid \revise{inductive} invariant within the lattice. This approach offers several key advantages:
\begin{itemize}
    \item It is \emph{declarative} and parameterized by the abstract domain $\mathcal{A}$ and the partial order $\sqsubseteq$, providing a precise characterization of BIIs without invoking symbolic abstraction as a primitive.
    \item  It is \emph{constructive}, enabling the use of algorithmic strategies that can circumvent the convergence limitations of standard fixed-point iteration.
\end{itemize}

\smallskip
\noindent \textbf{Template Domains over Bit-Vectors}.
We focus on \emph{template-based abstract domains} over bit-vector variables. 
They admit a uniform vector representation, allowing us to treat invariant synthesis as an optimization problem.

\begin{definition}[Template Domain]
\label{def:template-domain}
Let $X = (x_1, \dots, x_n)$ be the vector of program variables, and $F(X) = [f_1(X), \dots, f_m(X)]^\top$ be a fixed vector of template functions, where each $f_i$, linear or non-linear, maps the variables to a bit-vector of width $\mathsf{bw}_i$. The \emph{template domain} $\mathcal{D}_T$ consists of all abstract elements $A$ defined by a pair of constant bound vectors \revise{\delete{$L, U \in \mathbb{BV}^m$} $L,U \in \prod_{i=1}^{m}\mathbb{BV}_{\mathsf{bw}_i}$}:
$$
    A(X) \iff L \le F(X) \le U,
$$
where the inequality holds component-wise (i.e., $\forall i. \, l_i \le f_i(X) \le u_i$).
\end{definition}

\revise{Note the vector $F(X)$ is chosen by the client as part of the abstract domain. The client-defined functions determine the semantics of the template: each $f_i$ may have its own specified bit-width and arithmetic semantics, e.g., modular overflow and signedness. In the following sections, each template function $f_i$ and its bounds are interpreted as unsigned bit-vectors, and comparisons use the standard unsigned numerical order.}

The template domain $\mathcal{D}_T$ then forms a finite-height lattice ordered by the tightness of the interval bounds. The partial order $\sqsubseteq$ is defined such that $A \sqsubseteq A'$ if and only if its bounds are contained within those of $A'$(i.e., $L' \le L$ and $U \le U'$). Naturally, the lattice rank function that serves as our optimization target can be expressed as the sum of the interval widths, \revise{\delete{$i.e.,$} interpreted over mathematical integers, i.e.,} $rank_{\mathcal{D}_T}(A) =\sum_{i=1}^{m} (u_i - l_i + 1)$.
According to Definition~\ref{def:template-domain}, each inequality $l_i \le f_i(X) \le u_i$ constrains only the scalar projection defined by $f_i$. 
We will exploit \revise{\delete{this component-wise independence attribute} the component-wise nature of the optimization target} by computing the optimal bounds $(l_i,u_i)$ for each dimension $i$\revise{\delete{separately}}.

\smallskip
\noindent \textbf{Solving the BII Problem}.
Our formulation can be viewed as a domain-specific instance of the Optimization Modulo Theories (OMT) problem~\cite{Sebastiani:2017:OMT:3080455.3080473,sebastiani2015optimathsat,vz,sebastiani2015optimization}, which generalizes SMT by determining models that minimize a given objective function.
\revise{\delete{However, standard OMT solvers are ill-equipped to handle the complex quantified constraints inherent to invariant synthesis.} Although this formulation is intuitively an OMT solvable instance, direct optimization requires solving an $\exists \forall$ formula: the template bounds are selected existentially while inductiveness must hold universally. This quantified coupling makes each optimization step expensive and provides no predictable finite refinement sequence.}

To address this gap, we propose a \textit{propose-and-refine} framework and two instantiations.
The first (\cref{sec:framework}) is a \textit{linear search} strategy over the lattice structure, iteratively refining the invariant by generating neighbors.
The second (\cref{sec:design}) is a \textit{bitwise greedy} strategy that leverages the domain's bit-level structure. 
Crucially, both strategies operate ``top-down'' in the lattice, ensuring that any intermediate result remains a sound invariant.

%% file: 4.framework.tex
\section{A Propose-and-Refine Framework for BII Synthesis}
\label{sec:framework}

The constraint-based formulation of BII synthesis naturally suggests a refinement-based approach: compute the strongest inductive invariant by iteratively tightening an initial over-approximation. We propose a \emph{propose-and-refine} framework that casts BII synthesis as a directed search over a lattice-structured abstract domain.

We formalize BII synthesis as a directed search over the lattice induced by the abstract domain $\mathcal{D}_T$. The procedure maintains a current abstract element $A$, initialized to the trivial invariant $\top$, and monotonically tightens $A$ while preserving inductiveness.
The framework is parameterized by two operators:
\begin{enumerate}
    \item $\textbf{Propose}(A)$: Generates a finite set of candidate abstract elements $\mathcal{C}$ that are strictly smaller than the current abstract element $A$ (i.e., $\forall C \in \mathcal{C}, C \sqsubset A$). These candidates represent potential directions for refinement (e.g., cutting a specific dimension).
    
    \item $\textbf{Refine}(A, \mathcal{C})$: Attempts to refine the current invariant $A$ using candidates in $\mathcal{C}$. It returns a pair $(A', \textit{stop})$, where $A'$ is the updated invariant (equal to $A$ if refinement failed, or tighter if successful) and $\textit{stop}$ is a boolean signal indicating whether the search should terminate immediately (if $A$ is proven to be the best inductive invariant).
\end{enumerate}

\begin{figure*}[t]
\centering
\begin{minipage}{0.42\textwidth}
\begin{algorithm}[H]
    \caption{Synthesizing BII via propose-and-refine}
    \label{alg:propose-refine}
    \small 
    \KwIn{Logic encoding $P(A, X, X')$}
    \KwOut{Best Inductive Invariant $A^{\star}$}

    \SetKwFunction{Propose}{Propose}
    \SetKwFunction{Refine}{Refine}
    \SetKwProg{Fn}{Function}{:}{}

    $A \gets \top$\;
    \tcp{Generate candidates}
    $\mathcal{C} \gets \Propose(A)$\;
    \While{$\mathcal{C} \neq \emptyset$}{
        \tcp{Attempt to refine A}
        $(A, stop) \gets \Refine(A, \mathcal{C})$\;
        \If{$stop$}{
            \tcp{Confirmed optimality}
            \textbf{break}\;
        }
        $\mathcal{C} \gets \Propose(A)$\;
    }
    \Return $A$\;
\end{algorithm}
\end{minipage}
\hfill
\begin{minipage}{0.57\textwidth}
\begin{algorithm}[H]
    \caption{A linear search strategy for BII}
    \label{alg:linear-search}
    \small
    \KwIn{Logic encoding $P(A, X, X')$}
    \KwOut{Best Inductive Invariant $A^{\star}$}

    \SetKwFunction{Propose}{Propose}
    \SetKwFunction{Refine}{Refine}
    \SetKwProg{Fn}{Function}{:}{}

    \Fn{\Propose{$A$}}{
        $\mathcal{C} \gets \emptyset$\;
        \ForEach{dimension $i \in \{1, \dots, m\}$}{
            \tcp{Propose tightening bounds}
            \tcp{\revise{if $l_i>u_i$ or overflowed, the candidate is $\bot$}}
            $C_{l} \gets A[l_i \leftarrow A.l_i + 1]$\;
            $C_{u} \gets A[u_i \leftarrow A.u_i - 1]$\;
            $\mathcal{C} \gets \mathcal{C} \cup \{C_{l}, C_{u}\}$\;
        }
        \Return $\mathcal{C}$\;
    }

    \Fn{\Refine{$A, \mathcal{C}$}}{
        \If{$\exists A' . (\bigvee_{C \in \mathcal{C}} A' \sqsubseteq C) \land \forall X, X' . P(A', X, X')$}{
            \Return $(A', \textbf{false})$\;
        }
        \Return $(A, \textbf{true})$\;
    }
\end{algorithm}
\end{minipage}
\end{figure*}

\begin{theorem}[\revise{Correctness of Algorithm~\ref{alg:propose-refine}}]
\label{thm:framework-correct}
\revise{For brevity, write $\mathsf{Inv}(A) \equiv \forall X,X'.\,P(A,X,X')$. Assume that the BII $A^\star$ exists and that $\mathsf{Inv}(\top)$ holds. Suppose:}
\begin{enumerate}
    \item \revise{if $\Refine(A,\mathcal C)=(A',\mathit{stop})$, then $\mathsf{Inv}(A')$ and $A'\sqsubseteq A$;}
    \item \revise{if $\Propose(A)=\emptyset$, then no $B\sqsubset A$ satisfies $\mathsf{Inv}(B)$;}
    \item \revise{if $\Refine(A,\mathcal C)=(A',\mathit{true})$, then no $B\sqsubset A'$ satisfies $\mathsf{Inv}(B)$.}
\end{enumerate}
\revise{Whenever Algorithm~\ref{alg:propose-refine} terminates, it returns $A^\star$. If the abstract domain has descending height $h$ below $\top$, then $A$ is strictly refined at most $h$ times.}
\end{theorem}

\begin{proof}
\revise{Initially, $A=\top$ is inductive. By Condition~1, every subsequent value of $A$ remains inductive and the sequence of current elements is descending. Let $A$ be the returned element. By Conditions~2 and~3, whichever termination condition is used, no strictly tighter inductive element exists below $A$. Since $A^\star$ is the BII and $A$ is inductive, $A^\star\sqsubseteq A$. If $A^\star\neq A$, then $A^\star\sqsubset A$, contradicting the termination condition. Hence $A=A^\star$. Finally, finite descending height bounds the number of strict refinements by $h$.}
\end{proof}


Algorithm~\ref{alg:propose-refine} presents the high-level procedure. The algorithm decouples the search strategy ($\Propose$) from the verification and update logic ($\Refine$). By embedding the check within $\Refine$, we allow the operator to flexibly handle solver feedback depending on the specific instantiation. Note that the algorithm is naturally \emph{anytime}, always maintaining a valid invariant $A$ that can be used for further verification even if it is terminated before it can fully compute the BII.

\begin{remark}
The framework is complete for abstract domains in which the initial over-approximation $\top$ has finite lattice rank with respect to the measure in \cref{subsec:bii:omt}, e.g., the bit-vector template domains considered in this paper. By contrast, domains like reals contain elements of infinite rank; for such domains, the framework cannot guarantee convergence to the BII, but still can be viewed as an anytime procedure for computing progressively tighter invariants.
\end{remark}

\smallskip
\noindent \textbf{Base Strategy: Linear Search}.
We first instantiate this framework with a \emph{Linear Search} strategy, as in Algorithm~\ref{alg:linear-search}, which corresponds to a coordinate descent on the lattice. This is the most intuitive approach, as it proposes abstract elements that cover all descendants of $A$.
\begin{itemize}
    \item $\textbf{Propose}(A)$: Generates the immediate neighbors of $A$ in the lattice. For a component $z_i$ with bounds $[l_i, u_i]$, it proposes candidates by incrementing lower bounds ($l_i+1$) or decrementing upper bounds ($u_i-1$).
    
    \item $\textbf{Refine}(A, \mathcal{C})$: Validates the candidate set $\mathcal{C}$ using an $\exists\forall$ solver. It returns a pair $(A', \textit{stop})$, where $A'$ is the updated invariant and $\textit{stop}$ is a boolean termination signal.
    \begin{itemize}
        \item If \textsc{Sat}, the solver returns a witness model $A'$. The operator returns $(A', \textit{false})$, effectively updating the invariant to $A'$ (where $A' \sqsubseteq C \sqsubset A$). The \textit{false} signal indicates that refinement was successful and the search should continue.
        \item If \textsc{Unsat}, it implies that no valid \revise{inductive} invariant exists within any of the proposed candidates. The operator returns $(A, \textit{true})$. The \textit{true} signal indicates that $A$ cannot be further refined, confirming it as the best inductive invariant and terminating the search.
    \end{itemize}
\end{itemize}

\begin{theorem}[\revise{Correctness and complexity of Algorithm~\ref{alg:linear-search}}]
\label{thm:linear-correct}
\revise{For the bit-vector template domain, Algorithm~\ref{alg:linear-search} returns the BII. It makes at most $1+\sum_{i=1}^{m}(2^{\mathsf{bw}_i}-1)$ solver queries; for uniform width $b$, this is $O(m2^b)$. Each query contains at most $2m$ candidate disjuncts.}
\end{theorem}

\begin{proof}
\revise{A satisfiable refinement query returns an inductive $A'\sqsubset A$, while an unsatisfiable query leaves $A$ unchanged; hence Condition~1 of Theorem~\ref{thm:framework-correct} holds. Every strict descendant $B\sqsubset A$ lies below an immediate candidate obtained by tightening one bound on which $B$ differs from $A$; therefore, if the disjunctive query is unsatisfiable, no strictly tighter inductive element exists, establishing Condition~3. Algorithm~\ref{alg:linear-search} never terminates through an empty candidate set, so Condition~2 is vacuous, and correctness follows from Theorem~\ref{thm:framework-correct}. For the query bound, let $\mu(A)=\sum_i(A.u_i-A.l_i)$. Every successful query decreases $\mu$ by at least one, while $\mu(\top)=\sum_i(2^{\mathsf{bw}_i}-1)$. Thus there are at most $\sum_i(2^{\mathsf{bw}_i}-1)$ successful queries, followed by at most one unsuccessful query.}
\end{proof}


\begin{example}
\label{exmp:linear}
Table~\ref{table:linear-trace} and Figure~\ref{fig:linear-example} demonstrate the execution of the Linear Search strategy (Algorithm~\ref{alg:linear-search}) for a 3-bit variable $x$.

\begin{table}[t]
\centering
\caption{Trace of the linear search strategy (Alg.~\ref{alg:linear-search}) on variable $x$.}
\label{table:linear-trace}
%
\begin{tabular}{|c|c|l|l|l|}
\hline
\multicolumn{2}{|l|}{$P(A, X, X')$} & \multicolumn{3}{l|}{$(x=5 \to A(X)) \land (A(X) \land x<6 \land x'=x+1 \to A(X'))$} \\ \hline
\multicolumn{2}{|l|}{Initialization} & \multicolumn{3}{l|}{$A^\star = [0, 7]$ (3-bit unsigned)} \\ \hline
\textbf{Iter} & \textbf{Current $A^\star$} & \textbf{Candidate Set $\mathcal{C}$} & \textbf{Result} & \textbf{Update Action} \\ \hline

1 & $[0, 7]$ 
  & $\{[1, 7], [0, 6]\}$ 
  & \textsc{Sat} ($A'=[0, 6]$) 
  & $A^\star \gets [0, 6]$ \\ \hline

2 & $[0, 6]$ 
  & $\{[1, 6], [0, 5]\}$ 
  & \textsc{Sat} ($A'=[3, 6]$) 
  & $A^\star \gets [3, 6]$ \\ \hline

3 & $[3, 6]$ 
  & $\{[4, 6], [3, 5]\}$ 
  & \textsc{Sat} ($A'=[5, 6]$) 
  & $A^\star \gets [5, 6]$ \\ \hline

4 & $[5, 6]$ 
  & $\{[6, 6], [5, 5]\}$ 
  & \textsc{Unsat} 
  & Terminate \\ \hline

\multicolumn{2}{|l|}{Result} & \multicolumn{3}{l|}{The best interval invariant is $[5, 6]$.} \\ \hline
\end{tabular}%
\end{table}

\begin{figure*}[t]
	\centering 
	\includegraphics[width=\linewidth]{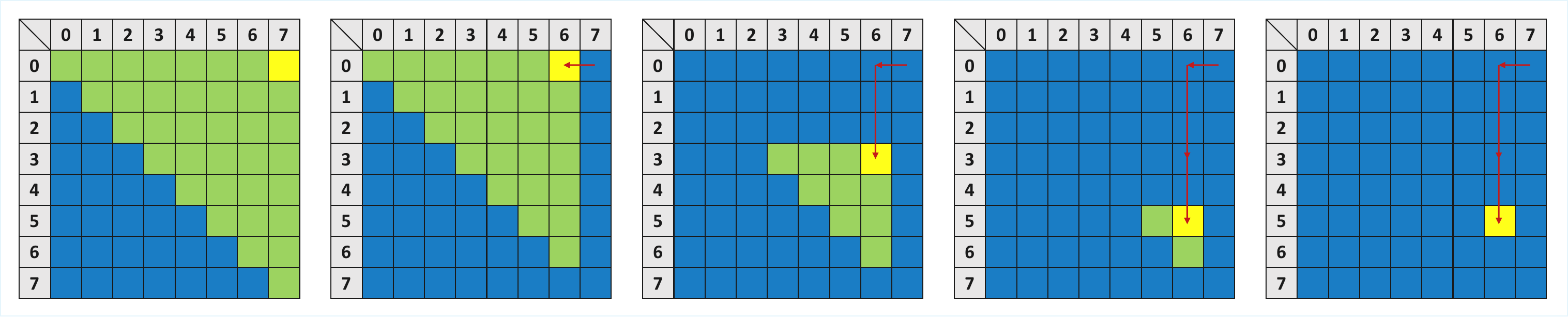}
	\caption{Illustration of the refining steps in Table~\ref{table:linear-trace}. Blue indicates unavailable lattices, yellow indicates the current solution, and green indicates the lattices tested in each step.}
    \label{fig:linear-example}
\end{figure*}

\end{example}

While this linear approach guarantees finding the BII by exhaustively proposing all possible solutions, its step-wise progression results in slow convergence.

%% file: 5.design.tex
\section{Bitwise Principled Refinement for BII Synthesis}
\label{sec:design}

The propose-and-refine framework boils down to a directed search over a lattice-structured abstract domain. However, the naive linear descent is inefficient: each candidate requires an expensive legality check, and fine-grained refinement converges slowly in large spaces.
We introduce a refinement strategy \revise{whose solver-query count is linear in the total bit-width of the template functions}. Inspired by \emph{binary lifting}~\cite{bender2004level}, we replace unit-step descent with exponential-step refinement, enabling efficient navigation of the search space (\cref{subsec:lifting}). Furthermore, we incorporate under-approximations derived to prune infeasible regions early and accelerate convergence (\cref{subsec:under-approx}).

\subsection{Proposing Strategy via Binary Lifting}
\label{subsec:lifting}
\noindent \textbf{Binary Lifting}.
Binary lifting is a classical algorithmic technique for efficiently answering queries over ordered structures, such as computing ancestors in trees. It exploits the binary representation of integers to reduce the number of search steps from linear to logarithmic. The method consists of a precomputation phase, in which jump pointers are constructed for powers of two, and a query phase, in which these pointers are used to incrementally refine a candidate solution. During a query, the algorithm iteratively considers jumps of size~$2^k$, starting from the largest relevant~$k$ and proceeding down to zero, updating the candidate when doing so yields a valid improvement.


\begin{example}
Consider finding the largest integer $x$ such that $x < 6$, starting with $x=1$.

\begin{enumerate}
    \item \emph{Phase 1: Precomputation}. We test increments starting from $0$ to find the largest exponent.
    \begin{itemize}
        \item Test $x + 2^0$. The result $2$ is valid, so $x \gets 2$.
        \item Test $x + 2^1$. The result $4$ is valid, so $x \gets 4$.
        \item Test $x + 2^2$. The result $8$ is invalid, so the precomputation phase ends here.
    \end{itemize}

    \item \emph{Phase 2: Query}. We test increments from the last valid exponent ($1$) down to $0$.
    \begin{itemize}
        \item Test $x + 2^1$. The result $6$ is invalid, so keep $x=4$.
        \item Test $x + 2^0$. The result $5$ is valid, so $x \gets 5$.
    \end{itemize}
    \item \emph{Result}. The procedure terminates with $x=5$.
\end{enumerate}
\end{example}

\smallskip
\noindent \textbf{Bitwise Greedy Strategy}.
Adapting binary lifting to the BII synthesis formulation produces a \textit{bitwise greedy} strategy: deciding the binary representation of the final result bit-by-bit, and unlike the standard binary lifting, it can skip some bit positions with the help of the refinement procedure.

This strategy operates on the principle that determining high-order bits first prunes the search space most aggressively. For a lower bound $l_i$ (initialized to $0$):


\begin{itemize}
    \item \textbf{Propose($A$):} At bit position $k$, we hypothesize that the optimal lower bound is at least large enough to have the $k$-th bit set. We propose a candidate $C$ where the $k$-th bit is forced to $1$. If the $k$-th bit is already set (due to a previous update), we skip to the next lower bit.
    \item \textbf{Refine($A, \mathcal{C}$):} We query the solver to check if a valid \revise{inductive} invariant exists within the candidate set $\mathcal{C}$.
    \begin{itemize}
        \item If \textsc{Sat}, the solver returns a witness model $A'$. We return $(A', \textit{false})$. This performs a \emph{greedy update}: we immediately adopt the tighter witness $A'$ as the new current invariant $A$. This often resolves multiple lower-order bits in a single step. The \textit{false} signal ensures the search continues until all remaining unresolved bits are resolved.
        \item If \textsc{Unsat}, the proposed bit configuration is invalid. We return $(A, \textit{false})$. This implies the $k$-th bit must remain at its current value (e.g., $0$ for a lower bound). Crucially, the stop signal remains \textit{false}, as the failure to set the $k$-th bit does not imply the search is complete; the algorithm must proceed to test the next bit ($k-1$).
    \end{itemize}
\end{itemize}

A symmetric logic applies to the upper bound $u_i$, where we attempt to force bits to $0$ (tightening the ceiling) from the highest bit position to the lowest bit position.
\revise{\delete{This strategy guarantees that the number of proposed candidates is at most $2*\sum_i{\mathsf{bw}_i}$, and that the number is halved at minimum during the validity check of each candidate.} Thus, every nonempty solver query resolves at least one previously unresolved bound bit.}

\revise{\delete{A candidate direction may occur in more than one query, so it is safer to count resolved bit positions and solver queries rather than generated candidate occurrences.}}

\begin{lemma}[\revise{\delete{Prefix preservation for bitwise refinement} Bit preservation and progress}]
\label{lem:bit-progress}
\revise{Assume that the BII $A^\star$ exists and that every solver query is exact. At every iteration of Algorithm~\ref{alg:bitwise-greedy}, $A^\star\sqsubseteq A$, and every lower- or upper-bound bit position already passed by \Propose agrees with the corresponding bit of $A^\star$. Every nonempty solver query passes at least one additional bit position.}
\end{lemma}

\begin{proof}
\revise{The claim holds initially because $A=\top$ and no position has been passed.}

\revise{If a query is \textsc{Sat}, it returns an inductive witness $A'\sqsubseteq A$. Since $A^\star$ is the BII, $A^\star\sqsubseteq A'$. Under the already fixed higher-order prefix, every tested or skipped lower-bound bit set to $1$ in $A'$ must also be $1$ in $A^\star$; otherwise $A^\star.l_i<A'.l_i$. The upper-bound case is symmetric. At least one candidate disjunct is satisfied, so at least one new position is passed.}

\revise{If a query is \textsc{Unsat}, no valid inductive element lies below any tested candidate. In particular, $A^\star$ lies below none of them, so each tested bit retains its current value. Since $A$ is unchanged, the next call to \Propose advances the corresponding position pointers.}

\revise{Thus the invariant is preserved, and every nonempty query passes at least one new position. This reasoning is inherently global: a \textsc{Sat} witness may simultaneously tighten multiple rows, and we never assume that feasibility factorizes across rows.}
\end{proof}

\begin{algorithm}[t]
\caption{Bitwise greedy instantiation}
\label{alg:bitwise-greedy}
\KwIn{Logic encoding $P(A, X, X')$}
\KwOut{Best Inductive Invariant $A^{\star}$}

\SetKwFunction{Propose}{Propose}
\SetKwFunction{Refine}{Refine}
\SetKwProg{Fn}{Function}{:}{}

\revise{
\textbf{global} $A_{last}\gets\bot$\;
\tcp*[l]{Shared with Algorithm~\ref{alg:advanced-refine}}
}

\Fn{\Propose{$A$}}{
    \textbf{static}
    $lpos \gets [\mathsf{bw}_1-1, \dots, \mathsf{bw}_m-1]$,
    $upos \gets [\mathsf{bw}_1-1, \dots, \mathsf{bw}_m-1]$
    \revise{\delete{, $A_{last} \gets \bot$}}\;

    $\mathcal{C} \gets \emptyset$\;
    \ForEach{dimension $i \in \{1, \dots, m\}$}{
        
        \If{$A = A_{last}$}{
            $lpos_i \gets lpos_i - 1$, $upos_i \gets upos_i - 1$\tcp*[l]{Move all position pointers if not refined}
        }

        $lpos_i \gets \max \{k \le lpos_i \mid (A.l_i)_k = 0\}$\tcp*[l]{Gets the highest unfixed unset(0) bit}
        
        \If{$lpos_i \ge 0$}{
            \tcp{Set the bit and clear lower bits}
            $C_{l} \gets A[l_i \leftarrow (A.l_i \ \& \ \sim ((1 \ll lpos_i) - 1)) \mid (1 \ll lpos_i)]$\;
            $\mathcal{C} \gets \mathcal{C} \cup \{C_{l}\}$\;
        }

        $upos_i \gets \max \{k \le upos_i \mid (A.u_i)_k = 1\}$\tcp*[l]{Gets the highest unfixed set(1) bit}
        \If{$upos_i \ge 0$}{
            \tcp{Clear the bit and set lower bits}
            $C_{u} \gets A[u_i \leftarrow (A.u_i \mid ((1 \ll upos_i) - 1)) \ \& \ \sim(1 \ll upos_i)]$\;
            $\mathcal{C} \gets \mathcal{C} \cup \{C_{u}\}$\;
        }
    }
    $A_{last} \gets A$\;
    \Return $\mathcal{C}$\;
}

\Fn{\Refine{$A, \mathcal{C}$}}{
    \If{$\exists A' . (\bigvee_{C \in \mathcal{C}} A' \sqsubseteq C) \land \forall X, X' . P(A', X, X')$}{
        \Return $(A', \textbf{false})$\;
    }
    \Return $(A, \textbf{false})$\;
}

\end{algorithm}

\begin{theorem}[\revise{Correctness and complexity of Algorithm~\ref{alg:bitwise-greedy}}]
\label{thm:bitwise-correct}
\revise{Let $W=\sum_{i=1}^{m}\mathsf{bw}_i$. Assuming that the BII exists and that every solver query is exact, Algorithm~\ref{alg:bitwise-greedy} returns the BII and \delete{generates at most $2W$ candidate bit directions, and} uses at most $2W$ solver queries. Each query contains at most $2m$ candidate disjuncts. For a uniform bit-width $b$, the query count is $O(mb)$.}

\end{theorem}

\begin{proof}
\revise{Every refinement result is inductive and no greater than the current element, so Condition~1 of Theorem~\ref{thm:framework-correct} holds. Algorithm~\ref{alg:bitwise-greedy} never returns $\mathit{stop}=\mathit{true}$, so Condition~3 is vacuous. By Lemma~\ref{lem:bit-progress}, every nonempty query passes at least one previously unresolved lower- or upper-bound bit. There are $2W$ such \delete{candidate directions} bit positions. When $\Propose(A)=\emptyset$, all positions have been passed, and every bound bit agrees with $A^\star$; hence $A=A^\star$, establishing Condition~2. Correctness follows from Theorem~\ref{thm:framework-correct}, and the same progress argument gives the $2W$ query bound. \delete{When all positions have been resolved, Lemma~\ref{lem:bit-progress} leaves no strictly tighter inductive element, so the current invariant is the BII by \cref{def:bii}.}}
\end{proof}

\begin{example}
Table~\ref{table:bitwise-trace} and Figure~\ref{fig:bitwise-example} demonstrate the bitwise greedy strategy (Algorithm~\ref{alg:bitwise-greedy}) on the same 3-bit variable $x$. The example shows that with a greedy update (at Bit 2), the algorithm can skip generating a candidate in the subsequent steps (at Bit 0).

\begin{table}[t]
\centering
\caption{Trace of the bitwise greedy strategy (Alg.~\ref{alg:bitwise-greedy}) on variable $x$.}
\label{table:bitwise-trace}
\resizebox{0.95\columnwidth}{!}
{
\begin{tabular}{|c|c|l|l|l|}
\hline
\multicolumn{2}{|l|}{$P(A, X, X')$} & \multicolumn{3}{l|}{$(x=5 \to A(X)) \land (A(X) \land x<6 \land x'=x+1 \to A(X'))$} \\ \hline
\multicolumn{2}{|l|}{Initialization} & \multicolumn{3}{l|}{$A^\star = [000_{(2)}, 111_{(2)}] ($[0, 7]$)$ (3-bit unsigned)} \\ \hline
\textbf{Iteration} & \textbf{Current $A^\star$} & \textbf{Candidate Set $\mathcal{C}$} & \textbf{Result} & \textbf{Update Action} \\ \hline

\multirow{2}{*}{1} & \multirow{2}{*}{$[000_{(2)}, 111_{(2)}]$} 
  & $C_l: [\textcolor{red}{1}00_{(2)}, 111_{(2)}]$ & \textsc{Sat} & \multirow{2}{*}{$A^\star \gets [101_{(2)}, 111_{(2)}]$} \\
  & & $C_u: [000_{(2)}, \textcolor{red}{0}11_{(2)}]$ & $A'=[101_{(2)}, 111_{(2)}]$ & \\ \hline

\multirow{2}{*}{2} & \multirow{2}{*}{$[101_{(2)}, 111_{(2)}]$} 
  & $C_l: [1\textcolor{red}{1}0_{(2)}, 111_{(2)}]$ & \multirow{2}{*}{\textsc{Unsat}} & \multirow{2}{*}{None} \\
  & & $C_u: [101_{(2)}, 1\textcolor{red}{0}1_{(2)}]$ & & \\ \hline

\multirow{2}{*}{3} & \multirow{2}{*}{$[101_{(2)}, 111_{(2)}]$} 
  & $C_l:$ \textit{skipped(bit is 1)} & \textsc{Sat} & \multirow{2}{*}{$A^\star \gets [101_{(2)}, 110_{(2)}]$} \\
  & & $C_u: [101_{(2)}, 11\textcolor{red}{0}_{(2)}]$ & $A'=[101_{(2)}, 110_{(2)}]$ & \\ \hline

\multicolumn{2}{|l|}{Result} & \multicolumn{3}{l|}{The best interval invariant is $[101_{(2)}, 110_{(2)}]$ ($[5, 6]$).} \\ \hline
\end{tabular}%
}
\end{table}

\begin{figure*}[t]
	\centering 
	\includegraphics[width=\linewidth]{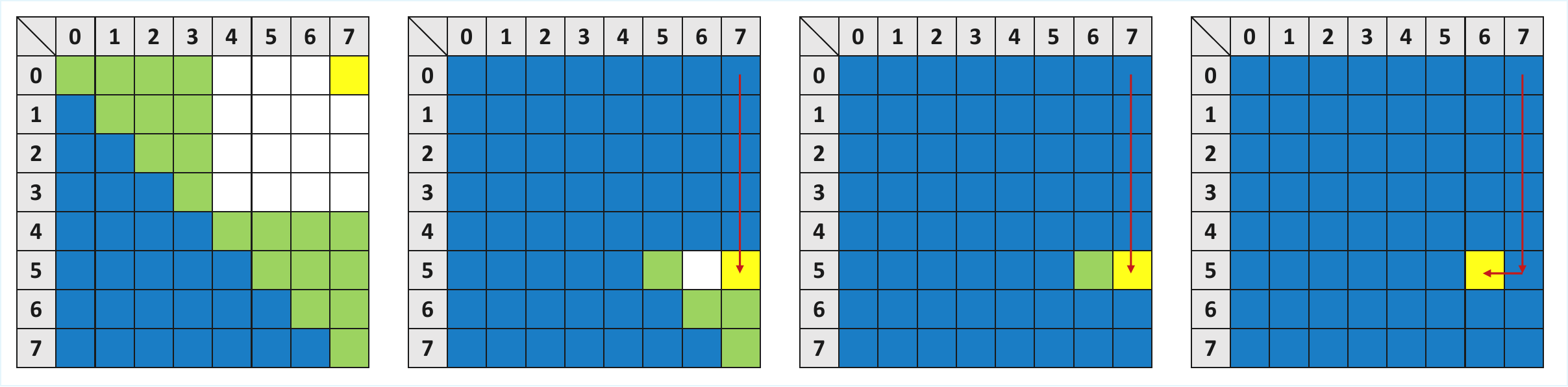}
	\caption{Illustration of the refining steps in Table~\ref{table:bitwise-trace}. Blue indicates unavailable lattices, yellow indicates the current solution, green indicates the lattices tested in each step, and white indicates the possible lattices temporarily not tested.}
    \label{fig:bitwise-example}
\end{figure*}

\end{example}

\subsection{Refining Strategy with the Under-Approximation}
\label{subsec:under-approx}
With the clear structural reduction from abstract elements to bounds, we are further able to detect an under-approximation from the $Refine(A,\mathcal{C})$ process. The algorithms above exploit only the \textsc{Sat} results. However, the \textsc{Unsat} results can also provide valuable information. By systematically tracking \textsc{Unsat} results, we can construct an under-approximation of the BII.

\smallskip
\noindent \textbf{Boundary Limits}.
To present the under-approximation, we introduce auxiliary variables for each dimension $i$, denoted as $lb_i$ and $ub_i$ for each component $z_i$, which we call the \textit{boundary limits}. Unlike the current bounds $l_i$ and $u_i$, which track the best valid invariant found so far, these auxiliary variables track the limits of the search space derived from failed checks. Specifically:
\begin{itemize}
    \item $lb_i$ represents the \revise{\delete{theoretical}} maximum possible value for the optimal lower bound $l^\star_i$.
    \item $ub_i$ represents the \revise{\delete{theoretical}} minimum possible value for the optimal upper bound $u^\star_i$.
\end{itemize}

Together with the current abstract element $A$, these variables constrain the optimal bounds within the ranges:
\[
A.l_i \le l^\star_i \le lb_i \quad \text{and} \quad ub_i \le u^\star_i \le A.u_i.
\]

We update these limits in the $\Refine(A, \mathcal{C})$ step:
\begin{itemize}
    \item \textsc{Unsat} (Pruning): If a proposal to tighten a lower bound to $C.l_i$ fails, it implies no valid invariant exists with a lower bound as high as $C.l_i$. Thus, the optimal lower bound $l^\star_i$ must be strictly less than $C.l_i$. We update the limit: $lb_i \leftarrow \min(lb_i, C.l_i - 1)$. 
    Symmetrically, if a proposal for an upper bound $c.u_i$ fails, we update $ub_i \leftarrow \max(ub_i, c.u_i + 1)$.
    
    \item \textsc{Sat} (Tightening): If a proposal succeeds, the solver returns a witness model $A'$. This model is a valid \revise{inductive} invariant, so we update the current best bounds $lb_i$ and $ub_i$ with respect to their counterparts $A'.u_i$ and $A'.l_i$ (i.e. $lb_i \leftarrow \min(lb_i,A'.u_i)$, $ub_i \leftarrow \max(ub_i,A'.l_i)$).

\end{itemize}


These boundary limits define the effective termination points of the search. Any candidate falling outside the active regions \([A.l_i, lb_i]\) or \([ub_i, A.u_i]\) can be immediately pruned. Moreover, these bounds can serve as auxiliary variables to constrain the search space in an SMT check. \revise{\delete{The interval \([lb_i, ub_i]\) forms a conservative approximation of the optimal bounds \([l_i, u_i]\)} The intervals \([lb_i,ub_i]\) record the search limits inferred from failed checks}, and the conjunction \(\bigwedge_i [lb_i, ub_i]\) yields an under-approximation of the target invariant \(A^\star\). We refer to this under-approximation as \(B\) in the remainder of the paper. \revise{These updates maintain $B\sqsubseteq A^\star$, as established in Theorem~\ref{thm:bounded-correct}.}

\smallskip
\noindent \textbf{Bounded Leap}.
As the search progresses, the gap between the current invariant $A$ and the under-approximation defined by the boundary limit narrows. When the search space becomes sufficiently constrained, a proposing strategy can be less efficient than a direct solve.

We introduce the \textit{bounded leap} strategy to exploit this regime. Whenever the under-approximation $B\neq\bot$ (i.e., $\forall i.\ lb_i \le ub_i$), we issue a single SMT query to refine the current invariant by finding an invariant within the constrained region:
\[
\exists A' . \forall X, X' . P(A', X, X') \land (B \sqsubseteq A' \sqsubset A).
\]

If the query is \textsc{Sat}, the result yields a strictly tighter invariant $A'$. If the query is \textsc{Unsat}, we conclude that $A$ is the best inductive invariant within the bounded region, enabling early termination.
\revise{After a successful bounded leap, we set $A_{last}\gets A'$ so that the next proposal advances past the directions rejected by the preceding standard query.}
The enhanced refinement procedure is presented in Algorithm~\ref{alg:advanced-refine}.

\begin{algorithm}[t]
\caption{Refinement with under-approximation}
\label{alg:advanced-refine}
\KwIn{Logic encoding $P(A, X, X')$}
\KwOut{Best Inductive Invariant $A^{\star}$}

\SetKwFunction{Propose}{Propose}
\SetKwFunction{Refine}{Refine}
\SetKwProg{Fn}{Function}{:}{}

\Fn{\Refine{$A, \mathcal{C}$}}{
    \textbf{static} $lb \gets [2^{\mathsf{bw}_i}-1, \dots]$, $ub \gets [0, \dots]$\;

    \tcp{1. Standard Check: Verify candidates}
    \If{$\exists A' . (\bigvee_{C \in \mathcal{C}} A' \sqsubseteq C) \land (\bigwedge_i \revise{\delete{l_i}A'.l_i} \leq lb_i) \land (\bigwedge_i ub_i \leq \revise{\delete{u_i}A'.u_i}) \land \forall X, X' . P(A', X, X')$}{
        \tcp{SAT: Witness found. Update limits using model $A'$.}
        \ForEach{dimension $i$}{
            
            $lb_i \gets \min(lb_i, A'.u_i)$, $ub_i \gets \max(ub_i, A'.l_i)$\;
        }
        \Return $(A',\textbf{false})$\;
    }

    \tcp{2. UNSAT: Prune limits based on failed candidates.}
    \ForEach{$C \in \mathcal{C}$}{
        \ForEach{dimension $i$}{
            \If{$C.l_i > A.l_i$}{
                $lb_i \gets \min(lb_i, C.l_i - 1)$\tcp*[l]{Attempt to tighten $l_i$ failed $\implies l^\star_i < C.l_i$}
            }
            \If{$C.u_i < A.u_i$}{
                $ub_i \gets \max(ub_i, C.u_i + 1)$\tcp*[l]{Attempt to tighten $u_i$ failed $\implies u^\star_i > C.u_i$}
            }
        }
    }

    \tcp{3. Bounded Leap: Try solving directly in the constrained region.}
    $B \gets \bigwedge_i [lb_i, ub_i]$\;
    \If{$B \neq \bot$}{
        \If{$\exists A' . B \sqsubseteq A' \sqsubset A \land \forall X, X' . P(A', X, X')$}{
            \revise{
            \ForEach{dimension $i$}{
                $lb_i \gets \min(lb_i,A'.u_i)$,
                $ub_i \gets \max(ub_i,A'.l_i)$\;
            }
            }

            \revise{
            $A_{last}\gets A'$\tcp*[l]{Consume the failed standard directions next}
            }
            \Return $(A',\textbf{false})$\;
        }
        \Else{
            \Return $(A, \textbf{true})$\tcp*[l]{No tighter invariant exists}
        }
    }

    \Return $(A, \textbf{false})$\;
}

\end{algorithm}
\normalsize

\begin{theorem}[\revise{Correctness and complexity with under-approximation}]
\label{thm:bounded-correct}
\revise{Let $W=\sum_{i=1}^{m}\mathsf{bw}_i$. Assume that the BII $A^\star$ exists, that $\mathsf{Inv}(\top)$ holds, and that all solver queries are exact. Algorithm~\ref{alg:bitwise-greedy}, using the refinement procedure in Algorithm~\ref{alg:advanced-refine}, returns $A^\star$ and makes at most $4W$ solver queries.}
\end{theorem}

\begin{proof}
\revise{Initially, $B=\bot\sqsubseteq A^\star$. A successful standard or bounded-leap query returns an inductive $A'\sqsubset A$, while every other result leaves $A$ unchanged; hence Condition~1 of Theorem~\ref{thm:framework-correct} holds. The corresponding boundary-limit updates preserve $B\sqsubseteq A^\star$.}

\revise{If a standard query is \textsc{Unsat}, it excludes only candidate regions that contain no valid inductive element satisfying the boundary limits. Since $B\sqsubseteq A^\star$, these additional limits do not exclude $A^\star$. Consequently, the bit-preservation and progress argument of Lemma~\ref{lem:bit-progress} continues to hold. Thus, when $\Propose(A)=\emptyset$, Condition~2 holds.}

\revise{If a bounded-leap query is \textsc{Unsat} and $A^\star\sqsubset A$, then $A^\star$ itself would satisfy $B\sqsubseteq A^\star\sqsubset A$ and $\mathsf{Inv}(A^\star)$, contradicting unsatisfiability. Hence $A=A^\star$, establishing Condition~3. Correctness follows from Theorem~\ref{thm:framework-correct}.}

\revise{Lemma~\ref{lem:bit-progress} bounds the number of standard queries by $2W$. At most one bounded-leap query follows each failed \delete{bitwise candidate query} standard query, and the update of $A_{last}$ prevents the rejected directions from being retested. Therefore the total is at most $4W$.}
\end{proof}

\subsection{Summary}
\label{subsec:summary}
We have presented a propose-and-refine approach to BII synthesis over ordered abstract domains. The linear search instantiation serves as a simple and complete baseline, but its convergence is inherently incremental. The bitwise refinement strategy improves on this baseline by replacing unit-step descent with \revise{\delete{logarithmic} linear} refinement in the bit width, while boundary limits and bounded leap further exploit solver feedback and construct an under-approximation to prune infeasible regions and accelerate the remaining search.

The next section evaluates these design choices empirically and compares them with symbolic abstraction-based approaches for the same target domain~\cite{yao2021program,thakur2012bilateral,reps2004symbolic}.

%% file: 6.evaluation.tex
\section{Evaluation}
\label{sec:evaluation}

Here, we evaluate the presented algorithms by investigating the following research questions:
\begin{itemize}
    \item \textbf{RQ1}: How efficient are our algorithms compared to existing approaches (\cref{subsec:eval:performance})?
    \item \textbf{RQ2}: What factors influence the performance of the evaluated algorithms (\cref{subsec:eval:in-depth})?
    \item \textbf{RQ3}: What is the effect of the proposed strategies over the base implementation (\cref{subsec:eval:ablation})?
    \item \textbf{RQ4}: To what extent can BII improve the effectiveness of $k$-induction (\cref{subsec:eval:enhance})?
\end{itemize}

\noindent \textbf{Abstract Domains}.
Our approach applies to abstract domains expressible in the canonical form $L \leq F(X) \leq U$. \revise{\delete{We focus on interval and octagon abstractions as representative instances of this class, which offer a favorable balance between precision and performance.} Our evaluation instantiates three fixed template families: intervals, octagons, and sparse template polyhedra. For variables $x_1,\ldots,x_n$, the interval rows are $x_i$; the octagon rows additionally include $x_i+x_j$ and $x_i-x_j$ for $i<j$; and the sparse template-polyhedra rows additionally include canonical support-three forms $x_i+s_jx_j+s_kx_k$, where $i<j<k$ and $s_j,s_k\in\{-1,1\}$.}

\noindent \textbf{Benchmarks}.
Our evaluation uses benchmarks adapted from multiple established sources, including (1) LoopInvGen, which aggregates tasks from SyGuS-COMP, SV-COMP, and other verification literature (e.g., HOLA~\cite{dillig2013inductive})
and (2) the multi-phase benchmarks~\cite{riley2022multi}). 
To assess performance under bit-vector semantics, we instantiate benchmark variables as 32-, 64-, and 128-bit vectors.
\revise{At each selected width $b$, we translate original integer values (including template rows and bounds) into unsigned bit-vectors of width $b$. Thus, overflow is intentional and models program behavior with fixed-width integer types.}

To ensure a meaningful evaluation, \revise{\delete{we curated a subset of non-trivial benchmarks by retaining} each benchmark suite used in this paper retains} only instances in which more than half of the components in $F(X)$ have a non-trivial invariant (i.e., distinct from $[\min, \max]$). \revise{\delete{This process yielded 195 instances for the interval domain and 120 for the octagon domain.} This criterion yields the evaluation suites used in this paper: 195 interval instances, 120 octagon instances, and 69 sparse template-polyhedra instances. The criterion is meant to avoid tables dominated by trivial top invariants, since many of the translated benchmarks have trivial variable ranges.}

\smallskip
\noindent \textbf{Baselines}.
We implement the algorithms in this paper as a tool, \ToolName, which takes Constraint Horn Clause (CHC) files as input and outputs minimal inductive loop invariants. By default, we consider the Init and Inductiveness conditions to infer query-independent invariants (\cref{subsec:inv}). 

RQ1--RQ3 focus on \emph{exact BII synthesis} in a fixed abstract domain. Accordingly, we compare against symbolic-abstraction procedures that target BIIs. Table~\ref{tab:algorithms} summarizes the synthesis baselines used in RQ1--RQ3. \revise{For each benchmark, all evaluated synthesis algorithms receive the same template vector $F$, bit-widths, ordering semantics, and transition encoding; thus, they compute the BII in the same fixed abstract domain.}

RQ4 addresses a different question: whether BIIs improve \emph{downstream verification}. For this purpose, we compare against established verifiers with different proof mechanisms with a hybrid verifier that combines $k$-induction with EFBII(G) \revise{(see Table~\ref{tab:algorithms})}.

\begin{table}[t]
\centering
\caption{A summary of the evaluated algorithms.}
\label{tab:algorithms}
\resizebox{0.85\textwidth}{!}
{
\begin{tabular}{ l   l  }
\toprule
\textbf{Algorithm} & \textbf{Description} \\
\midrule
CIBII(BS)  & \cref{subsec:bii:exiting}, Algorithm~\ref{alg:bii:symabs} with the binary search approach~\cite{yao2021program} for \revise{\delete{$\widehat{f}(a)$} $f^\star(a)$}  \\
CIBII(Bi)  & \cref{subsec:bii:exiting}, Algorithm~\ref{alg:bii:symabs} with the state-of-the-art approach~\cite{thakur2012bilateral} for \revise{\delete{$\widehat{f}(a)$} $f^\star(a)$}  \\
EFBII(Lin) & \cref{sec:framework}, Algorithm~\ref{alg:linear-search} implemented on our framework\\
EFBII(G)   & \cref{sec:design}, Algorithm~\ref{alg:bitwise-greedy} + Algorithm~\ref{alg:advanced-refine} implemented on our framework\\
\bottomrule
\end{tabular}
}
\end{table}



\smallskip
\noindent \textbf{Environment}.
Our experiments are conducted on a machine equipped with an Intel(R) Xeon(R) Platinum 8176 CPU and 512 GB of RAM, running Ubuntu 22.04. A 60-second timeout is imposed for each synthesis task unless otherwise specified.

\subsection{The Overall Performance Comparison (RQ1)}
\label{subsec:eval:performance}

\begin{table}[t]
\caption{Algorithm details on the commonly solved instances of the four algorithms (all bit sizes combined).}
\label{tab:commonset_all}
\resizebox{0.85\textwidth}{!}
{
\begin{tabular}{l c c c c c}
\toprule
\textbf{Algorithm} & \textbf{Total (s)} & \textbf{Avg (s)} & \textbf{Total Checks} & \textbf{Avg Checks} & \textbf{Time/Check(ms)} \\
\midrule
\multicolumn{6}{c}{\textbf{Interval Domain (66 instances)}} \\
CIBII(BS) & 234.44 & 3.55 & 1,287,379 & 19505.74 & 0.18 \\
CIBII(Bi) & 217.78 & 3.30 & 633,592 & 9599.88 & 0.34 \\
EFBII(Lin) & 59.49 & 0.90 & 3,998 & 60.58 & 14.88 \\
EFBII(G) & 27.14 & 0.41 & 1,177 & 17.83 & 23.06 \\
\midrule
\multicolumn{6}{c}{\textbf{Octagon Domain (20 instances)}} \\
CIBII(BS) & 180.27 & 9.01 & 963,706 & 48185.30 & 0.19 \\
CIBII(Bi) & 123.13 & 6.16 & 269,137 & 13456.85 & 0.46 \\
EFBII(Lin) & 30.58 & 1.53 & 364.00 & 18.20 & 84.01 \\
EFBII(G) & 6.86 & 0.34 & 198.00 & 9.90 & 34.67 \\
\midrule
\multicolumn{6}{c}{\textbf{\revise{Sparse Template Polyhedra Domain (24 instances)}}} \\
\revise{CIBII(BS)} & \revise{100.77} & \revise{4.20} & \revise{397,266} & \revise{16552.75} & \revise{0.25} \\
\revise{CIBII(Bi)} & \revise{78.55} & \revise{3.27} & \revise{134,642} & \revise{5610.08} & \revise{0.58} \\
\revise{EFBII(Lin)} & \revise{93.62} & \revise{3.90} & \revise{1,267} & \revise{52.79} & \revise{73.89} \\
\revise{EFBII(G)} & \revise{47.03} & \revise{1.96} & \revise{304} & \revise{12.67} & \revise{154.70} \\
\bottomrule
\end{tabular}
}
\end{table}

\begin{figure*}[t]
	\centering 
	\begin{subfigure}{0.33\textwidth}
	\centering 
		\includegraphics[width=\linewidth]{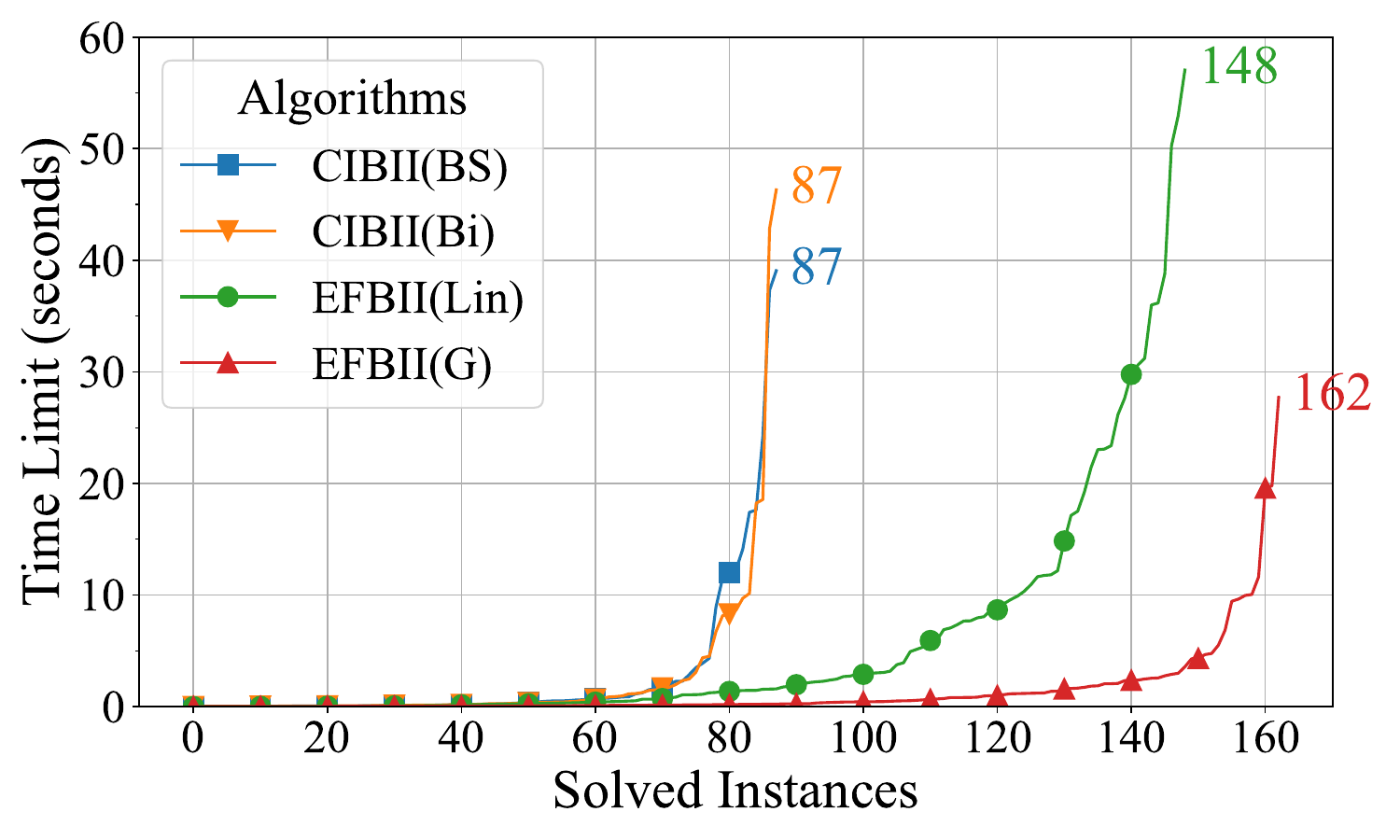}
		\caption{Interval domain}
        \label{fig:cactus:int}
	\end{subfigure}\hfil
	\begin{subfigure}{0.33\textwidth}
	\centering 
		\includegraphics[width=\linewidth]{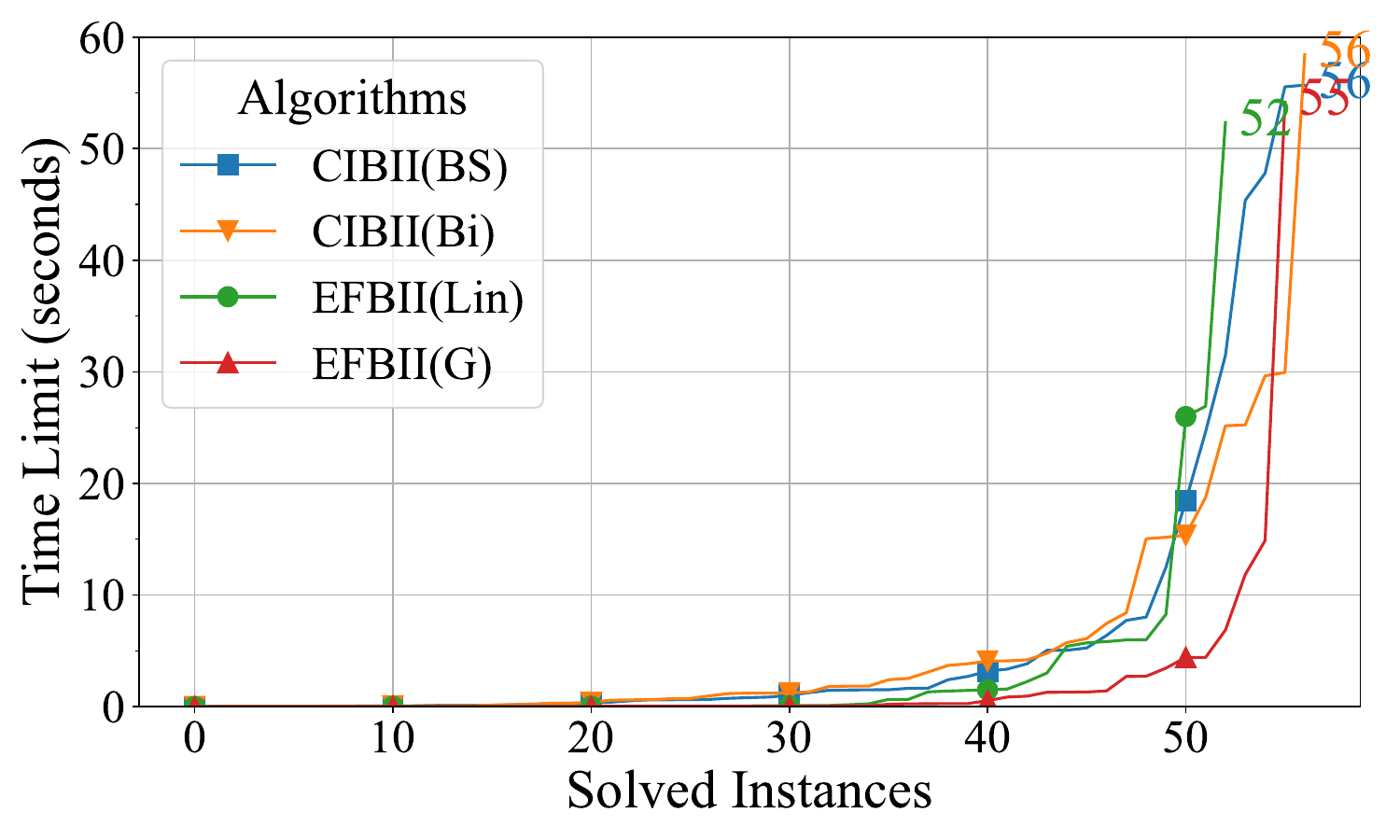}
		\caption{Octagon domain}
        \label{fig:cactus:oct}
	\end{subfigure}\hfil
	\begin{subfigure}{0.33\textwidth}
	\centering
		\includegraphics[width=\linewidth]{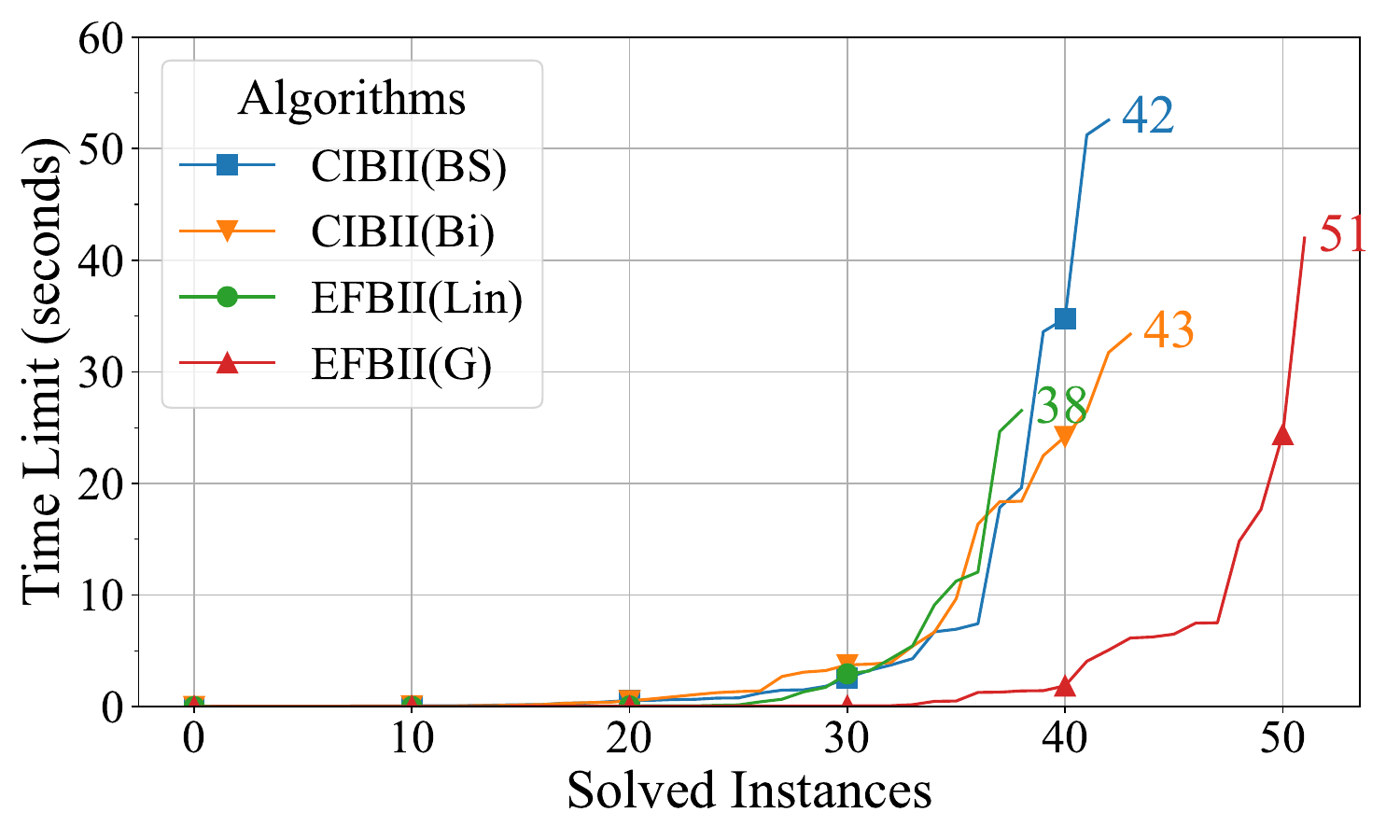}
		\caption{\revise{Template polyhedra}}
        \label{fig:cactus:poly}
	\end{subfigure}\hfil
	\caption{Cactus plot comparing the performance of algorithms on the \revise{\delete{interval and octagon domains} interval, octagon, and a sparse template polyhedra domain}, showing the cumulative number of instances solved within a time limit.}
    \label{fig:cactus}
\end{figure*}

Figure~\ref{fig:cactus} presents the cactus plots for the \revise{\delete{interval and octagon domains} interval, octagon, and an example sparse template-polyhedra domain}, illustrating the cumulative number of solved instances \revise{(i.e. instances that the BII is computed within the timeout)} over time. Table~\ref{tab:commonset_all} details the runtime statistics on the subset of instances solved by all algorithms.

\smallskip \noindent \textbf{The Symbolic Abstraction Approach}. The two chaotic iteration-based algorithms, CIBII(BS) and CIBII(Bi), demonstrate similar efficacy in terms of the number of solved instances. In both the interval and octagon domains, they solved the same number of instances before timing out.

While CIBII(Bi) requires fewer solver calls on average than CIBII(BS) (e.g., roughly half as many in the interval domain), both methods are constrained by the iterative nature of fixed-point computation. This limitation prevents them from scaling to harder instances in the interval domain, as evidenced by the sharp vertical asymptotes in Figure~\ref{fig:cactus:int}. Given its better performance, we use CIBII(Bi) as the primary baseline for the following comparisons.

\smallskip
\noindent \textbf{Symbolic Abstraction vs. Our Approach}.
The constraint-based approach, represented by EFBII(G), significantly outperforms the CIBII baselines in the interval domain and demonstrates superior efficiency in the octagon domain, despite a slightly lower total number of solvable instances.

\emph{Solvability}: In the interval domain, the advantage is decisive. EFBII(G) solves 162 instances, nearly doubling the solvability of the CIBII methods, which plateau at 87. The linear search strategy, EFBII(Lin), also performs well, solving 148 instances. In the octagon domain, the results are mixed. The symbolic abstraction methods achieve the highest solvability (56 instances), while EFBII(G) solves 55 instances, and EFBII(Lin) solves 52 instances. Although EFBII(G) solves one fewer instance than the baseline, Figure~\ref{fig:cactus:oct} reveals that it is significantly faster for the vast majority of cases, maintaining a lower time curve before reaching its limit. This result suggests that the algorithms favor different problem structures: CIBII methods are effective for relational constraints, where symbolic abstraction can exploit specific geometric properties, whereas EFBII(G) is more general, excelling where bit-level precision and rapid search pruning are required. \revise{In the fixed sparse template-polyhedra domain, EFBII(G) also achieves the highest solvability, solving 51 instances compared with 42 for CIBII(BS), 43 for CIBII(Bi), and 38 for EFBII(Lin).}

\emph{Efficiency}: The efficiency advantage of EFBII(G) is starkly evident on the subset of benchmarks solved by all algorithms (Table~\ref{tab:commonset_all}). In the interval domain (66 common instances), EFBII(G) requires only 27.14 seconds, achieving an $8.0\times$ speedup over CIBII(Bi), which takes 217.78 seconds. In the octagon domain (20 common instances), the gap widens further: EFBII(G) finishes in just 6.86 seconds, representing an $17.9\times$ speedup over CIBII(Bi) (123.13s). \revise{The same pattern holds in the fixed sparse template-polyhedra domain: on its 24 commonly solved instances, EFBII(G) reduces the number of checks from 134,642 for CIBII(Bi) to 304 and reduces runtime from 78.55s to 47.03s.}

A critical algorithmic trade-off is also revealed: although the bitwise greedy queries in EFBII(G) are computationally heavier ($\sim$23ms/check) than symbolic abstraction queries ($\sim$0.34ms/check), the logarithmic convergence of the bitwise strategy drastically reduces the volume of necessary checks. For instance, in the interval domain, EFBII(G) reduces the check count from 633,592 to just 1,177. This allows EFBII(G) to achieve superior overall runtime and successfully solve harder instances that iterative CIBII approaches fail to solve. \revise{Its higher time per check in the sparse template-polyhedra domain (154.70ms/check) indicates that the larger relational templates make each quantified query more expensive. These results concern the evaluated fixed sparse-template domain, rather than unrestricted polyhedral analysis.}

\subsection{An In-Depth Analysis of Algorithm Performance (RQ2)}
\label{subsec:eval:in-depth}

To understand the factors driving the performance differences in RQ1, we analyze the algorithms' sensitivity to the bit-width of program variables and the complexity of the target invariant on the interval domain.

\begin{figure*}[t]
     \centering
	\begin{subfigure}{0.25\textwidth}
	\centering 
		\includegraphics[width=\linewidth]{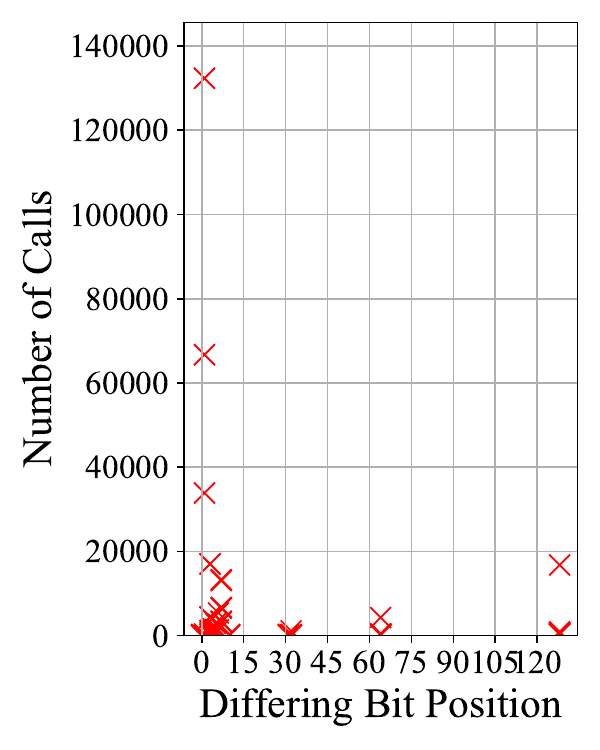}	
		\caption{CIBII(Bi)}
        \label{fig:DifferBitCall:Bi}
	\end{subfigure}\hfil
	\begin{subfigure}{0.25\textwidth}
	\centering 
		\includegraphics[width=\linewidth]{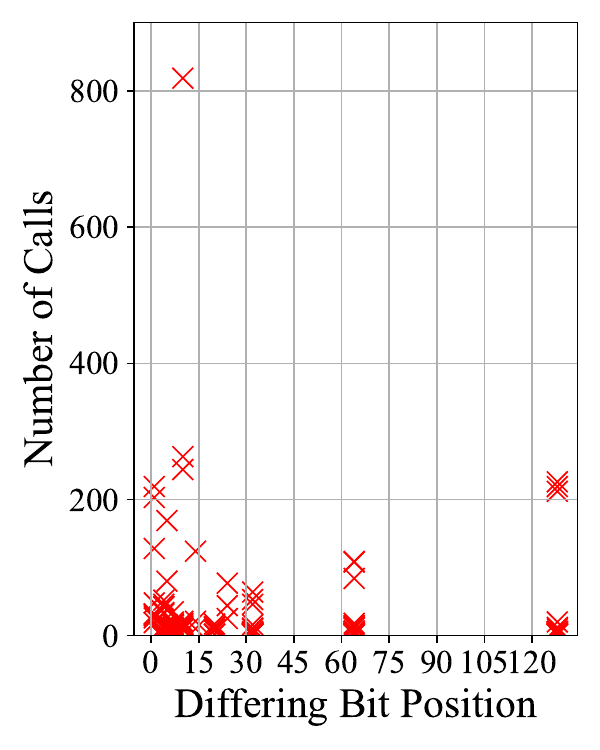}	
		\caption{EFBII(Lin)}
        \label{fig:DifferBitCall:Lin}
	\end{subfigure}\hfil
	\begin{subfigure}{0.25\textwidth}
	\centering
		\includegraphics[width=\linewidth]{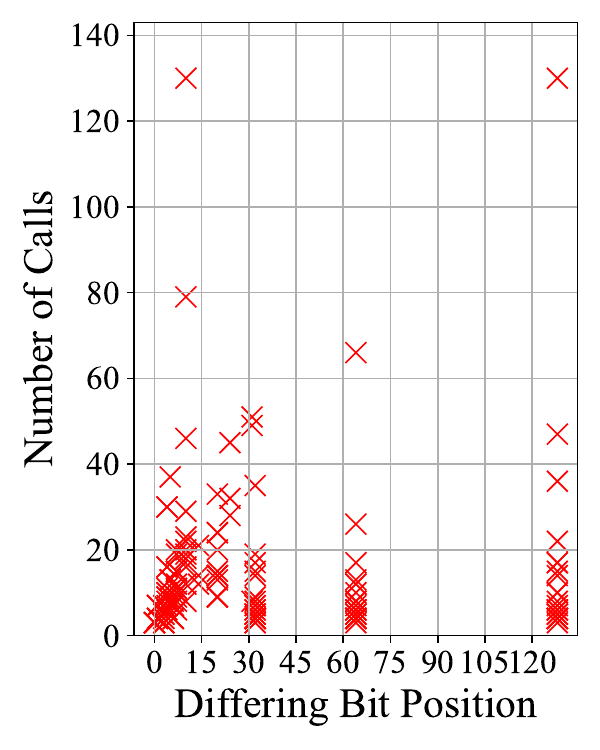}	
		\caption{EFBII(G)}
        \label{fig:DifferBitCall:G}
	\end{subfigure}\hfil 
      \caption{Count of solver calls on different invariant sizes (represented by the \revise{\delete{most significant differ bit} maximum differing bit position across all rows of the resulting BII}).}
      \label{fig:DifferBitCall}
\end{figure*}

\smallskip
\noindent \textbf{Sensitivity to Bit-Width}. Table~\ref{tab:width} details the performance of the algorithms on a subset of interval domain instances commonly solved across 32-bit, 64-bit, and 128-bit configurations.
\begin{itemize}
\item \textit{CIBII(Bi)} scales poorly due to an explosion in \textit{iteration count}. As the width increases from 32 to 128 bits, the number of symbolic abstraction checks more than quadruples (from $\sim$86,000 to $\sim$373,000). \revise{However, because chaotic iteration is relatively insensitive to the enlarged bit-vector search space, the average cost per check remains nearly stable (from 0.30ms to 0.38ms).} Consequently, the total runtime increases by $5.4\times$ (25.71s to 139.91s).

\item \textit{EFBII(Lin)} suffers from a compound effect: both the number of checks and the cost per check increase noticeably. The check count grows by $\sim3.4\times$ as the search space expands \revise{by $4\times$}, raising the average time per check from 3.31ms to 15.06ms. These two factors combine to produce \revise{\delete{the most} a} severe degradation, increasing total runtime by over $15\times$ (from 2.29s to 35.74s).

\item \textit{EFBII(G)} demonstrates the most robust \revise{\delete{algorithmic} scaling in the number of checks}. The number of checks increases by only $1.7\times$ (from 255 to 429) as the bit-width quadruples\revise{, compared with a $4\times$ increase on bit-width}. \revise{\delete{Although the cost per check increases significantly (reaching 64.42ms for 128-bit constraints), the drastically lower check count allows EFBII(G) to finish in 27.63s—nearly $5\times$ faster than CIBII(Bi) on 128-bit tasks.} However, the average cost per check increases from 6.84ms to 64.42ms, causing the total runtime to increase by $15.9\times$ (from 1.74s to 27.63s), compared with $5.4\times$ for CIBII(Bi). Nevertheless, its drastically lower check count allows EFBII(G) to remain approximately $5.1\times$ faster than CIBII(Bi) on 128-bit tasks.}
\end{itemize}

\smallskip
\noindent \textbf{Impact of Invariant Size}.
We further examine how the complexity of the target invariant affects performance by plotting solver calls against the \revise{\delete{position of the most significant differing bit, which serves as a proxy for the size of the invariant (Figure~\ref{fig:DifferBitCall}).} maximum differing-bit position across its rows (Figure~\ref{fig:DifferBitCall}). For a BII with rows $l_i \leq f_i(X) \leq u_i$, we define this metric as \revise{$\max_i \{k \mid (l_i \oplus u_i)_k = 1\} + 1$}, i.e., the maximum 1-based position at which a row's lower and upper bounds differ, and use it as a coarse proxy for the magnitude of the bounds. A value of zero means that every row is a singleton interval.} The results illustrate distinct behavioral profiles for each strategy. CIBII(Bi) proves highly unstable, with some instances requiring massive solver calls even for small invariant sizes, indicating that fixed-point iteration struggles to converge efficiently even on shallow invariants. In contrast, EFBII(Lin) reveals a linear dependence in which the search effort scales. EFBII(G) demonstrates remarkable stability: its call count remains uniformly low (under 140) regardless of the invariant size. This visually confirms the\revise{\delete{logarithmic}} efficiency of the bitwise greedy strategy, which navigates to the target invariant with a predictable and minimal number of steps, regardless of the magnitude of the bounds.

\begin{table}[t]
\centering
\caption{Performance of 3 main algorithms on different bit-widths (60s timeout). The statistics are based on the \revise{\delete{25} 22} commonly solved instances across all methods and bit-widths in the interval domain.}
\label{tab:width}
\resizebox{0.8\textwidth}{!}
{
\begin{tabular}{l l c c c c}
\toprule
\textbf{Algorithm} & \textbf{Width} & \textbf{\#Checks} & \textbf{Avg Checks} & \textbf{\#Time} & \textbf{Time/Check(ms)} \\
\midrule
\multirow{3}{*}{CIBII(Bi)} & 32 & 85,763 & 3898.32 & 25.71 & 0.30 \\
& 64 & 175,228 & 7964.91 & 52.15 & 0.30 \\
& 128 & 372,601 & 16936.41 & 139.91 & 0.38 \\
\midrule
\multirow{3}{*}{EFBII(Lin)} & 32 & 691 & 31.41 & 2.29 & 3.31 \\
& 64 & 933 & 42.41 & 13.10 & 14.04 \\
& 128 & 2,374 & 107.91 & 35.74 & 15.06 \\
\midrule
\multirow{3}{*}{EFBII(G)} & 32 & 255 & 11.59 & 1.74 & 6.84 \\
& 64 & 268 & 12.18 & 6.25 & 23.32 \\
& 128 & 429 & 19.50 & 27.63 & 64.42 \\
\bottomrule
\end{tabular}
}
\end{table}

\subsection{Ablation Study of the EFBII Algorithm (RQ3)}
\label{subsec:eval:ablation}

To quantify the impact of our design decisions, we evaluate the evolution of the propose-and-refine framework on the interval domain by comparing the linear search baseline (EFBII(Lin)) against three progressive stages of the bitwise greedy strategy: the foundational bitwise greedy strategy (EFBII(G-Base)), the intermediate version incorporating boundary limits for pruning (EFBII(G-BL)), and the fully optimized algorithm equipped with bounded leap (EFBII(G)).

The results are visualized in Figure~\ref{fig:binlifts}, while the statistics on 123 commonly solved instances are summarized in Table~\ref{tab:BinLiftsRes}.

\begin{figure*}[t]
    \begin{minipage}[c]{0.54\textwidth}
        \centering
        \captionof{table}{Comparing EFBII(G) across different configurations. Time and check statistics reflect 123 commonly solved instances.}
        \label{tab:BinLiftsRes}
        \resizebox{\linewidth}{!}{
           \begin{tabular}{lccccc}
            \toprule
            \textbf{Algorithm} &      \textbf{Total(s)} &       \textbf{Avg(s)} &
            \textbf{\# Calls} &       \textbf{Avg Calls} &       \textbf{Time/Call(ms)} \\
            \midrule
            EFBII(Lin) & 95.7 & 0.78 & 4,714 & 38.3 & 20.31 \\
            EFBII(G-Base) & 731.4 & 5.95 & 10,308 & 83.8 & 70.95 \\
            EFBII(G-BL) & 352.0 & 2.86 & 9,766 & 79.4 & 36.04 \\
            EFBII(G) & 87.1 & 0.71 & 1,684 & 13.7 & 51.75 \\
            \bottomrule
            \end{tabular}
        }
    \end{minipage}\hfill
    \begin{minipage}[c]{0.44\textwidth}
        \centering
        \includegraphics[width=\linewidth]{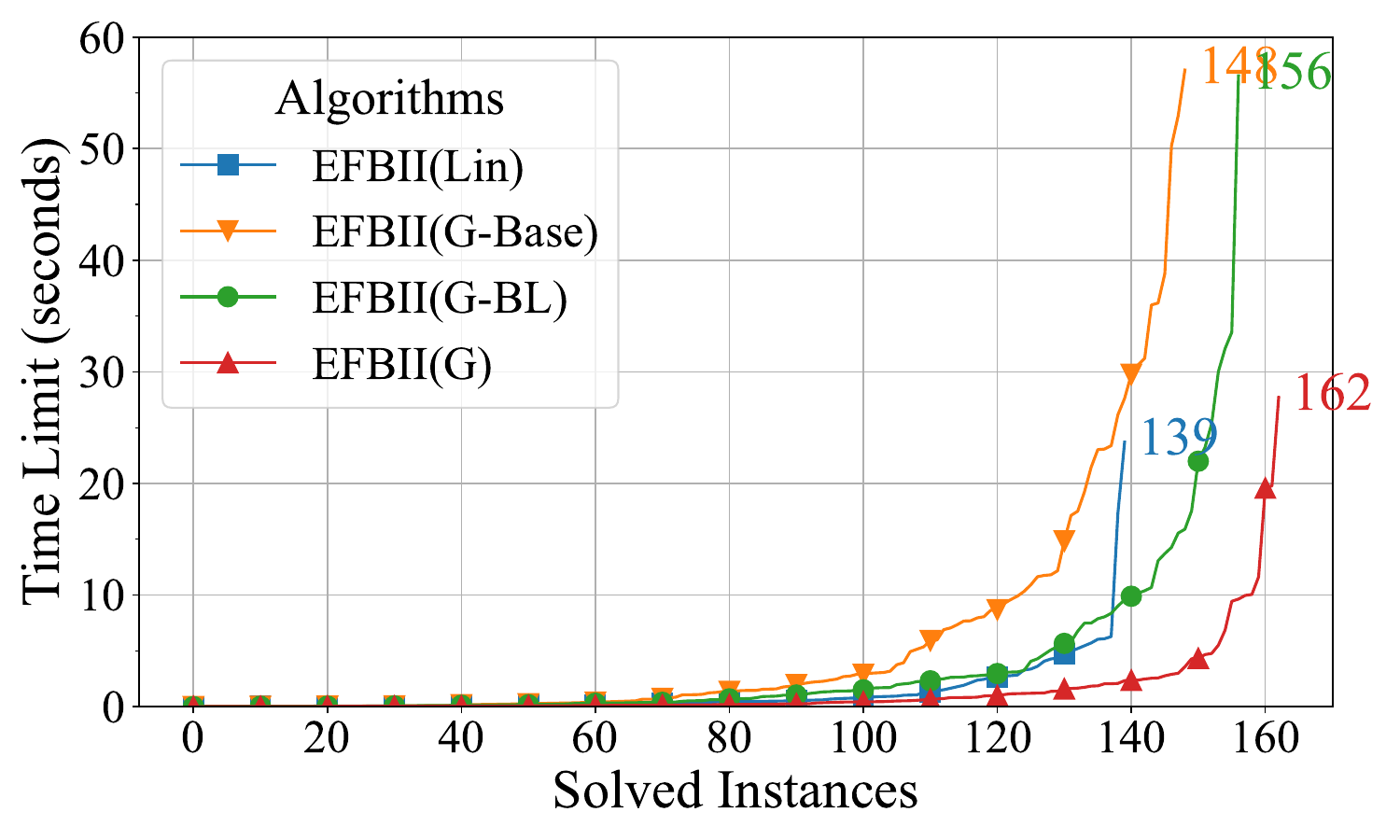}
        \captionof{figure}{Cumulative invariant inference time of EFBII variations.}
        \label{fig:binlifts}
    \end{minipage}
\end{figure*}

\smallskip \noindent \textbf{The Efficiency-Solvability Trade-off (Linear vs. G-Base)}. The comparison between EFBII(Lin) and EFBII(G-Base) reveals a distinct trade-off. In terms of solvability, EFBII(G-Base) is more robust, solving 148 instances compared to 139 for EFBII(Lin). This confirms that the bitwise strategy is necessary to handle the harder, large-width instances where linear search fails. However, on the set of commonly solved (and generally simpler) instances, EFBII(G-Base) is inefficient. It requires 731.4 seconds to solve the same subset that EFBII(Lin) clears in just 95.7 seconds. This is due to the heavy overhead of bitwise queries: the complex bit-masking constraints increase the solver time per check to 70.95ms, compared to just 20.31ms for simple bounds in EFBII(Lin). Without optimization, EFBII(G-Base) overcomplicates simple problems and performs significantly worse than a naive linear descent.

\smallskip\noindent \textbf{Impact of Boundary Limits}.
Incorporating boundary limits (under-approximation) in EFBII(G-BL) significantly improves both metrics. It pushes solvability to 156 instances and cuts the total runtime on the common set by more than half (731.4s to 352.0s). The pruning effectively filters out the most computationally expensive invalid queries, reducing the average time per check from 70.95ms to 36.04ms.

\smallskip \noindent \textbf{Impact of Bounded Leap}. Bounded-leap optimization is the critical factor that enables the bitwise strategy to outperform the linear baseline on the commonly solved set. It effectively mimics the efficiency of linear search on simple sub-problems, reducing the total number of calls from 9,766 to just 1,684. As a result, the fully optimized EFBII(G) achieves the best of both worlds: it achieves the highest solvability (162 instances) and is the fastest algorithm on the common set (87.1s), finally surpassing EFBII(Lin) (95.7s) in raw speed.

This study confirms that while the bitwise strategy provides the scalability for hard problems, it incurs a heavy penalty on simple ones. The bounded leap optimization is strictly necessary to bridge this gap, ensuring the solver remains efficient across the entire complexity spectrum.

\subsection{Enhancing $k$-Induction using Auxiliary Invariants (RQ4)}
\label{subsec:eval:enhance}
To evaluate the practical benefits of our approach, we conduct an experiment that integrates our best-performing algorithm, EFBII(G), with a $k$-induction engine. We assess whether supplying Best Inductive Invariants (BII) improves verification performance in two scenarios: first, by feeding pre-computed invariants of our algorithm to a standard $k$-induction loop; and second, by creating a hybrid solver that interleaves EFBII(G) steps with $k$-induction. \revise{The experiments use the full set of 195 benchmarks on the interval domain.}

\smallskip
\noindent \textbf{Assisting $k$-Induction via Partial and Optimal Results}.
\revise{We call the result of one refinement iteration of EFBII(G) a partial invariant: it is a sound inductive interval invariant produced before the search has established optimality, and it may be weaker than the BII. The optimal invariant is the BII obtained after the search terminates.} We supply both partial \revise{\delete{(derived from a single iteration of EFBII(G))}} and optimal \revise{\delete{(fully computed)}} invariants on the interval domain to a $k$-induction engine running with a 30-second timeout. The results, summarized in Table~\ref{tab:kind}, demonstrate that the provided invariants significantly strengthen verification capability.

\begin{itemize}
\item \textit{Increased Proof Rate}. Raw $k$-induction proves 125 instances. Supplying an optimal invariant increases the number of solved cases to 137, reducing the number of unprovable instances to 58—a 17\% reduction in failed proofs. Partial invariants also yielded a benefit, solving 2 more instances than the raw $k$-induction.

\item \textit{Lower Inference Depth}. The optimal invariants enable the engine to prove properties at shallower induction depths. For example, at $k=16$, the raw engine solves 121 cases, while the engine equipped with optimal invariants solves 133. The extra precision is useful when the downstream prover cannot arbitrarily increase the number of iterations.

\item \textit{Time Reduction}. The auxiliary invariants provide substantial speedups. The optimal invariants achieve speedups ranging from 1.04$\times$ to 1.19$\times$ across different depths. Notably, the partial invariants also deliver comparable speedups, highlighting the practical value of \revise{\delete{EFBII (G) 's} EFBII(G)'s} ``anytime'' nature.

\end{itemize}




\noindent \textbf{Comparison with other Verifiers}.
To further assess the utility of BIIs, we build a hybrid solver that alternates between 16 steps of $k$-induction and one \revise{\delete{interval-refinement step using EFBII(G)} EFBII(G) propose-and-refine iteration}. We denote this solver by ``$k$-EFBII(G)'' and compare it against four baselines under the same 60-second timeout: raw $k$-induction, Z3's Spacer engine (IC3/PDR-based), Eldarica~\cite{hojjat2018eldarica} (a CEGAR-based CHC solver that uses bit-vector interpolation~\cite{backeman2021interpolating} for refinement), and clause2inv\cite{Cao25Clause2Inv} (a state-of-the-art LLM-assisted verify engine). At this timeout, Eldarica achieves the highest final coverage with 142 solved instances, followed closely by $k$-EFBII(G) with 137. The hybrid closes most of the gap to the strongest baseline, reaching 96.5\% of Eldarica's solved count.

\emph{Time Analysis}. The hybrid dominates the front of the cactus curve. Within 1 second, $k$-EFBII(G) already solves 131 instances, compared with 124 for plain $k$-induction, 85 for Spacer, 13 for Eldarica, and none for clause2inv. In contrast, Eldarica obtains much of its advantage later in the run, increasing from 126 solved instances at 5 seconds to 141 at 20 seconds and 142 at 60 seconds. These data indicate that BII-guided refinement primarily improves \emph{time-to-proof}. Importantly, the hybrid never loses an instance that plain $k$-induction can prove within 60 seconds; the gain is strictly monotone.

\emph{Overlap Analysis}. The hybrid and Eldarica solve 130 common instances within 60 seconds, and in 129 of these instances, $k$-EFBII(G) is faster. Nevertheless, Eldarica still solves 12 instances that the hybrid misses; conversely, the hybrid uniquely solves 7 instances that Eldarica does not prove. Analysis on these instances suggests genuine complementarity: interpolation-driven CEGAR remains stronger on structurally difficult cases, while exact interval BIIs are highly effective on loops with strong numeric regularity. Notably, $k$-induction and $k$-EFBII(G) are also the only methods that solve any nonlinear benchmark, thereby discharging one NIA instance that other tools fail to prove.

\begin{figure*}[t]
    \begin{minipage}[c]{0.58\textwidth}
    \centering
    \captionof{table}{The impact of using EFBII(G) to assist $k$-induction. The Speedup is calculated from commonly solved instances.}
    \label{tab:kind}
    \resizebox{\linewidth}{!}{
        \begin{tabular}{lcccccc}
        \toprule
        \multirow{2}{*}{\textbf{Invariant Type}} & \multicolumn{4}{c}{\textbf{Proved cases in $k$ steps}} & \multirow{2}{*}{\textbf{Unprovable}} \\
        & 1 & 4 & 16 & 64 & \\
        \midrule
        Raw $k$-induction & 90 & 106 & 121 & 125 & 70 \\
        \midrule
        Partial Inv. & 90 & 106 & 123 & 127 & 68 \\
        Speedup & 1.034$\times$ & 1.196$\times$ & 1.173$\times$ & 1.148$\times$ & -- \\
        \midrule
        Optimal Inv. & 90 & 112 & 133 & 137 & 58 \\
        Speedup & 1.035$\times$ & 1.190$\times$ & 1.178$\times$ & 1.157$\times$ & -- \\
        \bottomrule
        \end{tabular}
    }
    \end{minipage}\hfill
    \begin{minipage}[c]{0.4\textwidth}
        \centering
        \includegraphics[width=\linewidth]{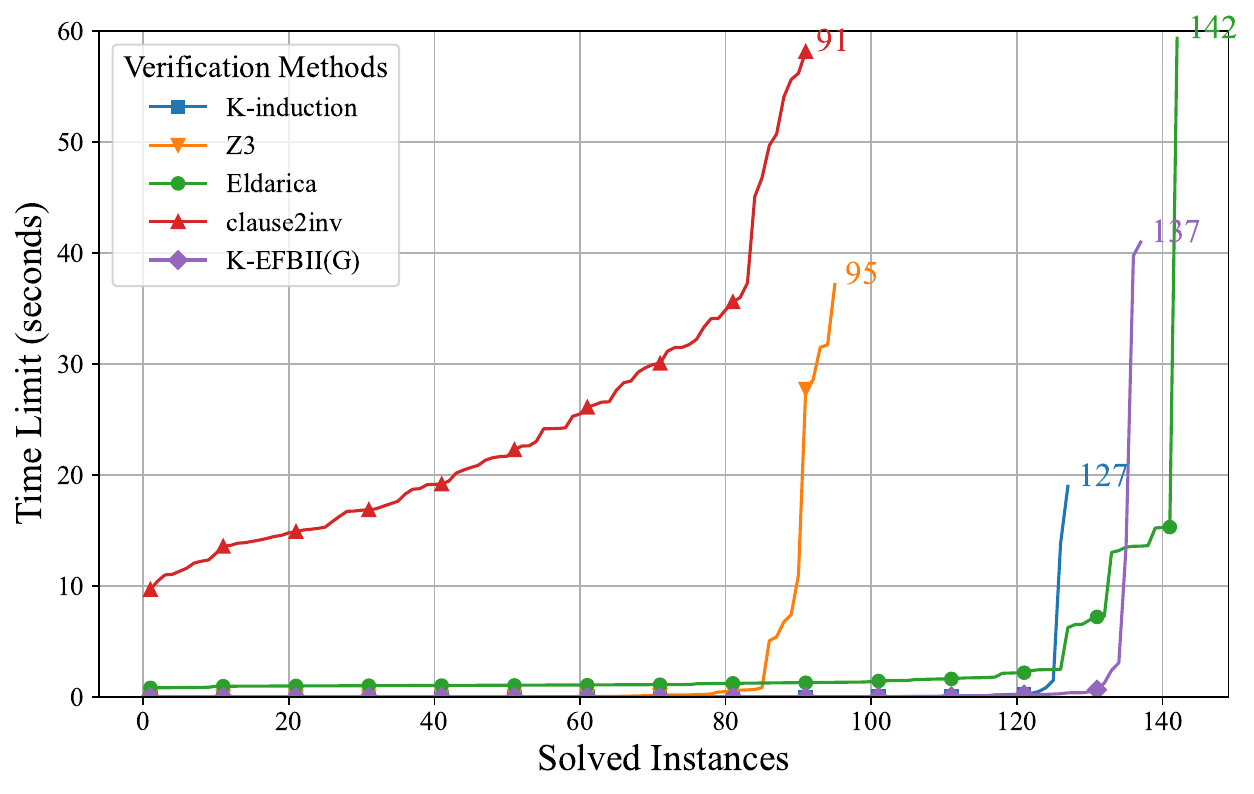}
        \captionof{figure}{Verification performance of the mixed $k$-induction and EFBII(G) strategy compared to other methods.}
        \label{fig:verify}
    \end{minipage}
\end{figure*}

\subsection{Discussions}
\label{sec:discuss}

\noindent \textbf{Applicability of Best Inductive Invariants}.
Best Inductive Invariants (BII) provide a principled approach for computing minimal inductive invariants within a given abstract domain.
In abstract interpretation, \revise{\delete{BII refines key operations, such as semantic reduction and domain combination.} BIIs can make the results of semantic reduction and domain combination more precise.}
Beyond abstract interpretation, BII can enhance other verification techniques, including the generation of auxiliary invariants for $k$-induction, as demonstrated in our evaluation. While prior work has explored the role of invariants in $k$-induction~\cite{beyer2015boosting}, existing approaches often rely on non-optimal invariants. 
Furthermore, the optimality of BII is particularly advantageous in compositional verification~\cite{gupta2008automated,calcagno2009compositional},
where imprecise invariants can propagate across components.

\smallskip
\noindent \textbf{Generalization to More Abstract Domains}. 
While our evaluation focuses on the \revise{\delete{interval and octagon domains} interval, octagon, and sparse template-polyhedra domains}, the proposed techniques can be generalized to a broader class of abstract domains. In the parity domain, for example, BII computes minimal congruences to track variable parities, thereby improving reasoning about modular arithmetic. Similarly, applying BII to zones, octagons, and affine relation domains over bit-vector arithmetic~\cite{elder2014abstract} could enhance the precision of analyzing modular linear constraints in low-level programs. In bit-level domains such as ``known bits,'' BII can infer precise must-information about individual bits (e.g., definite 0, definite 1, or unknown), further refining bitwise reasoning.

\smallskip
\noindent \textbf{Optimality in Guess-and-Check Approach}. 
Many invariant synthesis techniques follow a guess-and-check paradigm: candidate invariants are guessed — such as by instantiating templates with concrete values — and checked iteratively. Most existing methods typically do not attempt to find the best invariants. 
Recent work has explored similar ideas for considering optimality ~\cite{DBLP:journals/pacmpl/KalitaMDRR22,DBLP:journals/pacmpl/ParkDR23}, but it uses two independent processes to verify soundness and optimality. 
In comparison, our framework integrates synthesis and optimality checking into a unified loop. A modular alternative would decouple these concerns by iteratively generating candidates and verifying their optimality, but may lack guidance for the candidate-generation process.

%% file: 7.related.tex
\section{Related Work}
\label{sec:related}

\noindent \textbf{Best Inductive Invariants}.
Research on best inductive invariants (BII) can be broadly divided into two categories: (1) studies that focus on synthesizing the best abstract transformer (BAT)~\cite{Graf97,Regehr04,reps2004symbolic,Yorsh04,brauer2011transfer,King10,Monniaux10,thakur2012bilateral,thakur2012method}, which indirectly solves the BII problem and (2) research that directly addresses the BII problem itself~\cite{Flanagan01,Yorsh06,Garoche12}.
\revise{The optimization perspective on invariant computation also includes policy iteration~\cite{costan2005policy}, strategy iteration~\cite{gawlitza2007precise}, and convex-optimization formulations~\cite{gawlitza2012abstract}. These methods motivate treating invariant inference as an optimization problem, whereas our setting is specifically BII synthesis on finite bit-vector template domains.}
Existing approaches to BAT can be broadly categorized into two classes. 
The first class employs SMT-based iterative algorithms. Within this class, some techniques adopt a ``bottom-up'' approach, iteratively constructing a sequence of increasingly weaker implicants until one is entailed by $\varphi$~\cite{reps2004symbolic}. In contrast, other techniques follow a ``top-down'' strategy, generating a sequence of progressively stronger implicants until no further strengthening is possible~\cite{thakur2012method,thakur2012bilateral}.
The second class of approaches reformulates the problem into other automated reasoning tasks, such as OMT solving~\cite{yao2021program,li2014symbolic} or quantifier elimination~\cite{brauer2011transfer}. 

While the BII problem is closely related to the BAT problem, solving the former does not necessarily require solving the latter. For instance, our approach avoids using BAT synthesizers in its sub-procedures. 
The computation of BIIs is also linked to the problem of achieving completeness in abstract interpretation.
\citet{giacobazzi1997completeness,giacobazzi2000making} provide a constructive framework for characterizing completeness in abstract interpretation. In this context, BIIs (when they exist) can serve as formal witnesses of completeness. However, our work does not address the problem of deciding whether BIIs exist for a given program and abstract domain. For the problem, we refer the readers to~\cite{giacobazzi2025best} for a more thorough discussion.

\smallskip
\noindent \textbf{Constraint-based Invariant Generation}.
Constraint-based invariant generation, also known as the template-based approach, formulates invariant inference as a constraint-solving problem over unknown parameters within a predefined template.  This approach has been successfully applied to generate both linear invariants~\cite{colon2003linear,sankaranarayanan2005scalable,gupta2009invgen} and non-linear invariants~\cite{kapur2006automatically,chatterjee2020polynomial,chen2015counterexample}.  
Our work is inspired by the approach but differs in several dimensions.
First, whereas prior methods aim to enumerate all valid invariants or derive sufficient invariants for verification, we focus on synthesizing the best inductive invariant.
Second, we target bit-vector programs, whereas most existing techniques are designed for integer or real arithmetic. Prior approaches often rely on domain-specific reductions, such as Farkas' lemma for linear arithmetic~\cite{sankaranarayanan2005scalable} or Ackermann's reduction for uninterpreted functions~\cite{beyer2007invariant}, which do not directly extend to bit-vector semantics.
Finally, most prior work is restricted to affine programs, where loop guards and variable updates are affine functions. In contrast, our approach handles nonlinear constructs. 
\revise{Constraint-solving perspectives on program analysis have also been surveyed by Gulwani et al.~\cite{pldi08bv}; related logic-based invariant-generation procedures include~\cite{kahsai2011instantiation,garoche2013incremental}.}
\revise{The relation to template-polyhedral domains is also direct: prior work studies generalized templates, template polyhedra with additional structure, and bilinear optimization~\cite{templatepoly2011generalizing,templatepoly2017twist,templatepoly2019bilinear}.}

\smallskip
\noindent \textbf{Analysis of Modular Arithmetic}. 
This discrepancy between mathematical integers and finite-precision integers has motivated significant research into abstract domains tailored for bit-vector arithmetic~\cite{gange15,mine:hal-00748094,Sharma2017,Simon2007}. 
Below, we summarize key approaches that address this issue.
Tools such as Astrée~\cite{blanchet2002design} and cccheck~\cite{fahndrich2010static} focus on detecting expressions that are guaranteed to avoid overflow or underflow while issuing warnings for expressions that may be unsafe. 
The wrapped interval domain~\cite{gange15} precisely represents overflow and underflow by modeling the cyclic nature of bit-vector arithmetic.
Recently, \citet{yoon2023inductive} optimize loop-free program synthesis by combining
unsigned interval, signed interval, and bitwise abstractions.
These methods rely on instruction-level abstract interpretation and do not produce optimal abstract transformers.
In contrast, symbolic abstraction provides a framework for computing optimal abstract transformers and has been successfully applied across a range of domains, including intervals~\cite{regehr2006deriving, brauer2011transfer}, sets~\cite{brauer2010automatic}, affine relations~\cite{elder2014abstract}, octagons~\cite{sharma2017sound}, and polyhedra~\cite{sharma2017sound, yao2021program}. 
Our algorithms do not rely on standard chaotic iteration or symbolic abstraction at each refinement step; they avoid computing the best abstract transformers over loop-free fragments as subprocedures. \revise{For downstream verification, our use of BII as auxiliary facts is related to invariant-strengthened $k$-induction~\cite{beyer2015boosting,brain2015safety,rocha2015model}. Specifically, Brain et al.~\cite{brain2015safety} also obtain logarithmic dependence on the finite value range via SMT-assisted binary search on one row. In contrast, our algorithm fixes the bound bits from high to low in joint-solver queries across all rows, so a single query can refine several rows simultaneously.}

\smallskip
\noindent \textbf{\revise{Optimal Program Synthesis}}.
\revise{Optimal program synthesis has been studied most extensively in the programming-by-example (PBE) setting. Prior work formalizes optimality via explicit cost functions~\cite{Bornholt2016, Feser2015, Schkufza2013}, allowing users to bias synthesis toward programs with desirable structural properties, e.g., minimal size.
Several approaches instead adopt a probabilistic notion of optimality. Menon et al.~\cite{menon2013pbe}, for example, define an optimal program as one with maximum likelihood under a probabilistic context-free grammar conditioned on the examples. This view prioritizes candidates that are statistically plausible given the observed input–output behavior.
In contrast, our notion of optimality is derived from the underlying abstract domain rather than from client-specified objectives. This abstraction-centric formulation decouples the synthesis procedure from application-specific cost models, enabling a uniform and reusable approach.}

%% file: 8.conclu.tex
\section{Conclusion}
\label{sec:conclusion}

In this paper, we have revisited the problem of synthesizing best inductive invariants and presented a formulation that avoids invoking best-abstract-transformer computation over loop-free fragments as a primitive and derives two algorithms from it: a linear-search procedure and a bitwise greedy method \revise{\delete{with logarithmic convergence} with a solver-call count linear in total template bit-width (equivalently, logarithmic in the finite value-space size of the template rows)}. An empirical evaluation over different abstract domains shows substantial performance improvements over symbolic-abstraction baselines. Furthermore, we demonstrate the practical applicability of our approach by integrating it with $k$-induction.
